\documentclass{ws-ijgmmp}

\usepackage{xcolor}
\usepackage[verbose,hypertexnames=false]{hyperref}
\hypersetup{colorlinks=false,allbordercolors=blue,pdfborderstyle={/S/U/W 1}}

\begin{document}


%
\catchline{}{}{}{}{}
%

\title{Delta functions on twistor space and their sign factors \\ 
}

\author{Jun-ichi Note
}

\address{Research Institute of Science \& Technology, College of Science and Technology, 
\\
Nihon University, Tokyo 101-8308, Japan
\,
\\
\email{note.junichi20@nihon-u.ac.jp
} }



\maketitle


\begin{abstract}
When performing the Fourier transform of the scattering amplitudes in Yang-Mills theory from momentum space to real twistor space, we encounter sign factors that break global conformal invariance. Previous studies conjectured that the sign factors are intrinsic in the real twistor space corresponding to the split signature space-time; hence, they will not appear in the complex twistor space corresponding to the Lorentzian signature space-time. In this study, we present a new geometrical interpretation of the sign factors by investigating the domain of the delta functions on the real twistor space. In addition, we propose a new definition of delta functions on the complex twistor space in terms of the \v{C}ech cohomology group without any sign factors and show that these delta functions have conformal invariance. Moreover, we show that the inverse Fourier transforms of these delta functions are the scattering amplitudes in  Yang-Mills theory. Thus, the sign factors do not appear in the complex twistor space. 
\end{abstract}

\keywords{delta functions; twistor space; conformal invariance.}

\ccode{Mathematics Subject Classification 2020: 32L25}

\section{Introduction}

Studies on scattering amplitudes have been conducted using various methods and have progressed to the present \cite{BT}.
Twistor string theory is particularly impactful and uses twistor variables for analysis \cite{Wit}.
In this theory, amplitudes in $\mathcal{N}=4$ super Yang-Mills theory (SYM) undergo a Fourier transform from momentum space to twistor space.
This analysis was performed in the
real twistor space corresponding to the split signature space-time instead of  the complex twistor space corresponding to the Lorentz signature space-time, 
because real variables are in good agreement with the Fourier transform.  
The tree-level maximally helicity violating
(MHV) amplitudes in the real twistor space have sign factors which violate the invariance of global conformal  symmetry, particularly conformal inversion, despite a massless system \cite{MS,ACCK,KS}. However, these amplitudes have local conformal symmetry \cite{Wit}.
For example, 
the three-particle MHV amplitude, which is a component of the Britto-Cachazo-Feng-Witten (BCFW) recursion relation of tree amplitudes \cite{DH}, is represented as a product of the sign factor and delta function using twistor variables \cite{MS} or as a product of sign factors using both twistor and dual twistor variables \cite{ACCK}.
The sign factors have been conjectured not to appear in the complex twistor space because their information will be incorporated into the \v{C}ech cohomological structure \cite{MS}. Therefore, the appearance of sign factors is considered intrinsic to the real twistor space.    
Additionally, it is conjectured that delta functions on the complex twistor space will be constructed 
from a complex integral of complex delta functions which are defined by $\delta(z) = 1/(2 \pi i z)$ in terms of the \v{C}ech cohomology group without sign factors,  and the integral contour will be chosen as suitable \cite{MS}. 
However, the delta functions on the complex twistor space have been defined by the Dolbeault cohomology group in previous studies \cite{CSW}. 
This definition has been used in various recent  studies on the scattering amplitudes in twistor theory \cite{recent}.
The \v{C}ech cohomological definition of the delta functions has not been defined.
Therefore, the sign factors in the complex twistor space are unknown. However, investigating them is useful for calculating the scattering amplitudes in terms of the BCFW recursion relations in twistor space and would facilitate studies of the scattering amplitude for non-Abelian gauge theory. 

This paper presents a new geometrical interpretation of the sign factors in the real twistor space and a new definition of delta functions on the complex twistor space in terms of  the \v{C}ech cohomology group without any sign factors, for the first time.
First, we show that when we eliminate the sign factor in the definition of the delta function on the real twistor space $\mathbb{RP}^{3}$, it becomes a delta function on the three-dimensional sphere $S^{3} \approx \mathbb{RP}^{3} \times O(1)$, where $O(1):=\{-1,+1\}$. 
Additionally, for the collinear delta function which enforces the collinearity of the three twistors in the real twistor space and has a sign factor, we show that its domain is the real Grassmann manifold $G_{2,4}(\mathbb{R})$, which is the set of all two-dimensional planes in $\mathbb{R}^{4}$, and the domain of the collinear delta function without the sign factor is $G_{2,4}(\mathbb{R}) \times O(1)$. 
Hence, we show that the geometrical role of the sign factors is to divide the domain of the delta functions by $O(1)$.
Next, we propose a new definition of delta functions on $n$-dimensional complex space $\mathbb{C}^{n}$ as a representative element of the ($n-1$)-th \v{C}ech cohomology group on complex projective space $\mathbb{CP}^{n-1}$ with coefficients in the sheaf $\mathcal{O}(-n)$. This \v{C}ech cohomology group is denoted by $H^{n-1}(\mathbb{CP}^{n-1},\mathcal{O}(-n))$. We also propose a new definition of delta function on $\mathbb{CP}^{3}$ without any sign factors as a representative element of the \v{C}ech cohomology group $H^{2}(\mathbb{CP}^{2},\mathcal{O}(-3))$ by integrating the delta function on $\mathbb{C}^{4}$ by a specific parameter to impose the equivalence relation of $\mathbb{CP}^{3}$.
Here, we adopt dual complex twistor space
as $\mathbb{CP}^{3}$ to relate the delta functions to the scattering amplitudes.
We denote the delta function on $\mathbb{CP}^{3}$
by $\Delta_{m}^{(3)}(\bar{Z},\bar{X})$, where $m$ is an integer, and $\bar{X}$ and $\bar{Z}$ represent the points in $\mathbb{CP}^{3}$. This has the property of the delta function on $\mathbb{CP}^{3}$. Additionally, we define the delta function which enforces the collinearity of the three twistors $\bar{Z}_{1}$, $\bar{Z}_{2}$ and $\bar{Z}_{3}$
 in $\mathbb{CP}^{3}$ without any sign factors, 
denoted by $\Delta^{(2)}(\bar{Z}_{1},\bar{Z}_{2},\bar{Z}_{3})$.
This is constructed by specific parameters integration of  the delta function on $\mathbb{C}^{4}$  
and is a representative element of the \v{Cech} cohomology group $H^{1}(\mathbb{CP}^{1},\mathcal{O}(-2))$. These delta functions $\Delta_{m}^{(3)}(\bar{Z},\bar{X})$ and $\Delta^{(2)}(\bar{Z}_{1},\bar{Z}_{2},\bar{Z}_{3})$
are invariant under the local and global conformal transformations because they do not have the sign factors which violate conformal invariance,
and ingredient's delta function on $\mathbb{C}^{4}$ is invariant under the local and global conformal transformations.
Here, $\Delta_{m}^{(3)}(\bar{Z},\bar{X})$ is a solution of the differential equation of the twistor wave function \cite{Pen5}, because it possesses a degree of homogeneity in $m$ in $\bar{Z}$ and $(-m-4)$ in $\bar{X}$. The twistor wave functions correspond to massless free fields in the Minkowski space. Therefore, it is natural that $\Delta_{m}^{(3)}(\bar{Z},\bar{X})$ has conformal invariance.
Similarly, we define the delta function on the dual complex super twistor space $\mathbb{CP}^{3|4}$
by a specific parameter integration of a delta function on $\mathbb{C}^{4|4}$ which is defined by the product of the delta function on $\mathbb{C}^{4}$ and the Grassmann delta function. This is denoted by $\Delta^{(3|4)}(\bar{\mathcal{Z}}_{1},\bar{\mathcal{Z}}_{2})$, where $\bar{\mathcal{Z}}_{1}$ and $\bar{\mathcal{Z}}_{2}$ represent points in $\mathbb{CP}^{3|4}$. 
Additionally, we define the delta function $\Delta^{(2|4)}(\bar{\mathcal{Z}}_{1},\bar{\mathcal{Z}}_{2},\bar{\mathcal{Z}}_{3})$ which enforces the collinearity of the three supertwistors $\bar{\mathcal{Z}}_{1}$, $\bar{\mathcal{Z}}_{2}$ and $\bar{\mathcal{Z}}_{3}$ in $\mathbb{CP}^{3|4}$ by specific parameters integration of  the delta function on $\mathbb{C}^{4|4}$.
These delta functions 
$\Delta^{(3|4)}(\bar{\mathcal{Z}}_{1},\bar{\mathcal{Z}}_{2})$
and $\Delta^{(2|4)}(\bar{\mathcal{Z}}_{1},\bar{\mathcal{Z}}_{2},\bar{\mathcal{Z}}_{3})$ do not have any sign factors, therefore, they are invariant under the local and global super conformal transformations in the same manner as in the bosonic case. Furthermore, we show that the inverse Fourier transform of $\Delta^{(2|4)}(\bar{\mathcal{Z}}_{1},\bar{\mathcal{Z}}_{2},\bar{\mathcal{Z}}_{3})$ is the three-particle MHV amplitude for $\mathcal{N}=4$ SYM in the momentum superspace, and the inverse Fourier transform of the product of the two delta functions
$\Delta^{(3|4)}(\bar{\mathcal{Z}}_{1},\bar{\mathcal{Z}}_{2})\Delta^{(3|4)}(\bar{\mathcal{Z}}_{1},\bar{\mathcal{Z}}_{3})$
is the three-particle $\overline{\text{MHV}}$ amplitude for $\mathcal{N}=4$ SYM in the momentum superspace.
For this reason, the delta functions
$\Delta^{(3|4)}(\bar{\mathcal{Z}}_{1},\bar{\mathcal{Z}}_{2})$
and
$\Delta^{(2|4)}(\bar{\mathcal{Z}}_{1},\bar{\mathcal{Z}}_{2},\bar{\mathcal{Z}}_{3})$ are scattering amplitudes for $\mathcal{N}=4$ SYM in the dual complex super twistor space.
When the twistor space is complex space, the amplitudes are superconformally invariant because the sign factors do not appear. 

This paper is organized as follows. Section 2 reviews the delta functions on the real twistor space and their sign factors.  
Section 3 presents a new geometrical interpretation of sign factors.
Section 4 proposes a new definition of the delta functions on the dual complex (super) twistor space and shows that the inverse Fourier transforms of these delta functions are the scattering amplitudes in $\mathcal{N}=4$ SYM.  
Finally, Section 5 summarizes our study.

\section{Review of the delta functions on the real twistor space}

The delta function on the real twistor space 
$\mathbb{RP}^{3}$ is defined by using the following parameter integral \cite{MS}:
\begin{align}
\delta^{(3)}_{-n-4}(W,Y)
=\int_{-\infty}^{\infty} \dfrac{dt}{t^{1+n}} \text{sgn} (t) 
\delta^{(4)}(W-tY),
\label{d1}
\end{align}
where $W$ and $Y$ represent two points in $\mathbb{RP}^{3}$.
Hereafter, we denote the homogeneous coordinates of a point $W$ in $\mathbb{RP}^{3}$ by $W_{\alpha}$ ( $\alpha=0,1,2,3$ ).
In this definition, $\delta^{(3)}_{-n-4}(W,Y)$ has the support on the line $W_{\alpha} - t Y_{\alpha}=0$, that is, it is nonzero when $W$ and $Y$ are on the same line passing through the origin in $\mathbb{R}^{4}$.
Therefore, on the support, $W$ and $Y$
are the same points in $\mathbb{RP}^{3}$.
We see that $\delta^{(3)}_{-n-4}(W,Y)$ has the property
\begin{align}
\int_{\mathbb{RP}^{3}} f(W) \delta^{(3)}_{-n-4}(W,Y) D^{3} W = f(Y)
\end{align}
where $f(W)$ is a homogeneous function of degree $n$ and $D^{3}W$ is the volume element of $\mathbb{RP}^{3}$:
\begin{align}
D^{3}W := \dfrac{1}{6} \epsilon^{\alpha \beta \gamma \delta}
W_{\alpha} dW_{\beta} \wedge dW_{\gamma} \wedge dW_{\delta}.
\end{align}
Furthermore, $\delta^{(3)}_{-n-4}(W,Y)$ has a 
degree of
homogeneity $(-n-4)$ in $W$ and $n$ in $Y$: 
\begin{align}
\delta^{(3)}_{-n-4}(aW,Y) = \dfrac{1}{a^{n+4}} \delta^{(3)}_{-n-4}(W,Y),
\,\,\,\,\,
\delta^{(3)}_{-n-4}(W,bY) = b^{n} \delta^{(3)}_{-n-4}(W,Y).
\label{4}
\end{align}
From Eq. $(\ref{d1})$, we see that $\delta^{(3)}_{-n-4}(W,Y)$ is conformally invariant ($SL(4,\mathbb{R})$ invariant) because the integrand $\delta^{(4)}(W-tY)$
is manifestly conformally invariant.
Furthermore, a tilded $\delta$-function is defined by 
\begin{align}
\tilde{\delta}^{(3)}_{-n-4}(W,Y)
:= \int_{-\infty}^{\infty} \dfrac{dt}{t^{1+n}} 
\delta^{(4)}(W-tY) 
\label{tdelta}
\end{align}
without the sign factor.
This is associated with the delta function $\delta^{(3)}_{-n-4}(W,Y)$ through the relation:
\begin{align}
{\delta}^{(3)}_{-n-4}(W,Y) 
=\int_{-\infty}^{\infty} \dfrac{dt}{t^{1+n}}
\text{sgn} \left( \dfrac{W}{Y}\right) 
\delta^{(4)}(W-tY) 
= \text{sgn} \left( \dfrac{W}{Y}\right)
{\tilde{\delta}}^{(3)}_{-n-4}(W,Y), 
\end{align}
because $t=W/Y$ on the support of the integrand $\delta^{(4)}(W-tY)$ of $\delta^{(3)}_{-n-4}(W,Y)$.

The collinear  $\delta$-function which enforces the collinearity of the three twistors $W_{1}$, $W_{2}$, and $W_{3}$ in $\mathbb{RP}^{3}$ is defined by using the following parameter integral \cite{MS}:
\begin{align}
\delta^{(2)}(W_{1},W_{2},W_{3})
:=\int_{\mathbb{R}^{2}} \dfrac{ds}{s} \dfrac{dt}{t}
\text{sgn} (s) \text{sgn} (t) \delta^{(4)}
(W_{1}-sW_{2}-tW_{3}).
\label{w123}
\end{align}
This is a component of the three-particle MHV amplitude in twistor space.
Furthermore, the collinear $\tilde{\delta}$-function 
is defined by
\begin{align}
\tilde{\delta}^{(2)}(W_{1},W_{2},W_{3})
:=\int_{\mathbb{R}^{2}} \dfrac{ds}{s} \dfrac{dt}{t}
\delta^{(4)}
(W_{1}-sW_{2}-tW_{3})
\label{delta2}
\end{align}
without the sign factors.
We see that $\delta^{(2)}(W_{1},W_{2},W_{3})$ and $\tilde{\delta}^{(2)}(W_{1},W_{2},W_{3})$ are conformally invariant ($SL(4,\mathbb{R})$ invariant) because the integrand 
$\delta^{(4)}(W_{1}-sW_{2}-tW_{3})$ 
in Eqs. $(\ref{w123})$ and $(\ref{delta2})$
is manifestly conformally invariant. 
Here, we denote the coordinates of $W_{i}$ by $W_{i \alpha}:=(\lambda_{iA},\mu_{i}^{A^{\prime}})$ ($i=1,2,3$\,;\,$A=0,1$\,;\,$A^{\prime}=0^{\prime},1^{\prime}$ ) and define $\langle ij \rangle := \lambda_{iA} \lambda_{j}^{A}=\lambda_{iA} \epsilon^{AB} \lambda_{jB}
=\lambda_{i0} \lambda_{j1} - \lambda_{i1} \lambda_{j0}$.
In this notation, we have $s=\langle 31 \rangle/\langle 32 \rangle$ and
$t=\langle 12 \rangle/\langle 32 \rangle$ on the support of the integrand $\delta^{(4)}
(W_{1}-sW_{2}-tW_{3})$ in Eqs. $(\ref{w123})$ and $(\ref{delta2})$ because the equations $\lambda_{1A}-s\lambda_{2A}-t\lambda_{3A}=0$ are valid. 
By using this expression of $s$ and $t$, we have $st=\langle 31 \rangle \langle 12 \rangle/\langle 32 \rangle^{2}$, that is, $\text{sgn} (st) = \text{sgn} \left( \langle 31 \rangle \langle 12 \rangle \right)$.
From this expression and Eqs. $(\ref{w123})$ and $(\ref{delta2})$, we see that   
$\delta^{(2)} (W_{1},W_{2},W_{3})$ and $\tilde{\delta}^{(2)}(W_{1},W_{2},W_{3})$ are related by
\begin{align}
\delta^{(2)} (W_{1},W_{2},W_{3})
= \text{sgn} \left( \langle 31 \rangle \langle 12 \rangle \right)
\tilde{\delta}^{(2)}(W_{1},W_{2},W_{3}).
\label{3112}
\end{align}
Furthermore, we define the collinear ${\delta}$-function and $\tilde{\delta}$-function on the real super twistor space $\mathbb{RP}^{3|4}$ as follows:
\begin{align}
&
{\delta}^{(2|4)} \left( \text{W}_{1}, \text{W}_{2}, \text{W}_{3} \right)
:= \int_{\mathbb{R}^2} \dfrac{ds}{s} \dfrac{dt}{t}
\text{sgn}(s) \text{sgn}(t)
\delta^{(4|4)} \left( \text{W}_{1} - s \text{W}_{2} - t \text{W}_{3}   \right),
\label{deltaW123}
\\
&
\tilde{\delta}^{(2|4)} \left( \text{W}_{1}, \text{W}_{2}, \text{W}_{3} \right)
:= \int_{\mathbb{R}^2} \dfrac{ds}{s} \dfrac{dt}{t}
\delta^{(4|4)} \left( \text{W}_{1} - s \text{W}_{2} - t \text{W}_{3} \right),
\label{tildeW123}
\end{align}
where $\text{W}_{i}=\left( \lambda_{iA}, \mu_{i}^{A^\prime}, \chi_{ik} \right)$ ( $k=1,2,3,4$ ) denotes the bosonic coordinates $\lambda_{iA}$,  $\mu_{i}^{A^\prime}$, and
the fermionic coordinates $\chi_{ik}$ of $\mathbb{RP}^{3|4}$.
We see that ${\delta}^{(2|4)} \left( \text{W}_{1}, \text{W}_{2}, \text{W}_{3} \right)$ and $\tilde{\delta}^{(2|4)} \left( \text{W}_{1}, \text{W}_{2}, \text{W}_{3} \right)$ are superconformally invariant, because the integrand 
$\delta^{(4|4)} \left( \text{W}_{1} - s \text{W}_{2} - t \text{W}_{3} \right)$ 
in Eqs. $(\ref{deltaW123})$ and $(\ref{tildeW123})$
is manifestly superconformally invariant. 
In the same manner as Eq. $(\ref{3112})$,
we have the relationship
\begin{align}
{\delta}^{(2|4)} \left( \text{W}_{1}, \text{W}_{2}, \text{W}_{3} \right)
= \text{sgn} (\langle 31 \rangle \langle 12 \rangle)\,
\tilde{\delta}^{(2|4)} \left( \text{W}_{1}, \text{W}_{2}, \text{W}_{3} \right).
\end{align}
By carrying out the integration in Eq. $(\ref{tildeW123})$, we have
\begin{align}
(\ref{tildeW123})= \text{sgn} (\langle 23 \rangle)
\dfrac{\delta^{(2)} \left( \langle 23 \rangle \mu_{1}^{A^\prime} + \langle 31 \rangle \mu_{2}^{A^\prime} + \langle 12 \rangle \mu_{3}^{A^\prime} \right)
\delta^{(4)} \left( \langle 23 \rangle \chi_{1k} + \langle 31 \rangle \chi_{2k} + \langle 12 \rangle \chi_{3k} \right)
}{\langle 12 \rangle \langle 23 \rangle \langle 31 \rangle}.
\end{align}
Here, the Fourier transform of the three-particle MHV amplitude for $\mathcal{N}=4$ SYM in the split signature from the momentum superspace to the real super twistor space is represented as follows \cite{MS,ACCK}:
\begin{align}
\mathcal{A}_{\text{MHV}} \left( \text{W}_{1}, \text{W}_{2}, \text{W}_{3} \right)
&= i \text{sgn}(\langle 23 \rangle)\,
\tilde{\delta}^{(2|4)} \left( \text{W}_{1}, \text{W}_{2}, \text{W}_{3}\right)
\notag \\
&= i \text{sgn}(\langle 12 \rangle \langle 23 \rangle \langle 31 \rangle)\,
{\delta}^{(2|4)} \left( \text{W}_{1}, \text{W}_{2}, \text{W}_{3}\right).
\label{amhvsgn}
\end{align}
The sign factor in Eq. $(\ref{amhvsgn})$
ensures the cyclic symmetry of the amplitude and the property of particles interchange (for example, when particles 1 and 2 are interchanged, minus sign times the amplitude)
because ${\delta}^{(2|4)} \left( \text{W}_{1}, \text{W}_{2}, \text{W}_{3}\right)$ is symmetric under the interchange of variables.

Although the collinear $\delta$-function and $\tilde{\delta}$-function  
 in $\mathcal{A}_{\text{MHV}} \left( \text{W}_{1}, \text{W}_{2}, \text{W}_{3} \right)$ are superconformally invariant,
the sign factors $\text{sgn}(\langle 23 \rangle)$ and
$\text{sgn}(\langle 12 \rangle \langle 23 \rangle \langle 31 \rangle)$ are not invariant, particularly in conformal inversion. Therefore, $\mathcal{A}_{\text{MHV}} \left( \text{W}_{1}, \text{W}_{2}, \text{W}_{3} \right)$
is not superconformally invariant.
Here, the conformal inversion in the real super twistor space is a discrete element of $SL(4,\mathbb{R})$ whose square is the identity \cite{KS}.  This operation interchanges variables $\lambda$ and $\mu$ but does not change $\chi$ as follows: 
\begin{align}
I[\lambda_{A}] = \mu_{A^\prime}, \quad
I[\mu_{A^\prime}] = \lambda_{A}, \quad
I[\lambda^{A}] = -\mu^{A^\prime}, \quad
I[\mu^{A^\prime}] = -\lambda^{A}, \quad
I[\chi_{k}] = \chi_{k}.
\end{align}
Hence, for the conformal inversion, the sign factor transforms as
\begin{align}
I[\langle ij \rangle] 
= I[ \lambda_{iA} \lambda_{j}^{A} ]
= I[ \lambda_{iA} ] I[ \lambda_{j}^{A} ]
= \mu_{iA^\prime} \left( -\mu_{j}^{A^\prime} \right)
= - \mu_{iA^\prime} \mu_{j}^{A^\prime}
\ne \langle ij \rangle.
\end{align}
Therefore, the sign factor in Eq. $(\ref{amhvsgn})$ is not invariant under conformal inversion.
Furthermore, the Fourier transform of the three-particle $\overline{\text{MHV}}$ amplitude for $\mathcal{N}=4$ SYM in the split signature from the momentum superspace to the real super twistor space is represented as follows \cite{MS}:
\begin{align}
&\mathcal{A}_{\overline{\text{MHV}}} \left( \text{W}_{1}, \text{W}_{2}, \text{W}_{3} \right)
\notag \\
= 
&\text{sgn} \left( \left[ \dfrac{\partial}{\partial \text{W}_{2}} \dfrac{\partial}{\partial \text{W}_{3}} \right] \right)\,
\tilde{\delta}^{(3|4)} \left( \text{W}_{1}, \text{W}_{2} \right) \tilde{\delta}^{(3|4)} \left( \text{W}_{1}, \text{W}_{3} \right),
\label{tpdfists}
\end{align}
where the formal operator $\text{sgn} ([\partial_{\text{W}_{2}} \partial_{\text{W}_{3}}])$ is not invariant under super conformal transformation. Therefore, $\mathcal{A}_{\overline{\text{MHV}}} \left( \text{W}_{1}, \text{W}_{2}, \text{W}_{3} \right)$ is not superconformally invariant.

\section{New geometrical consideration of the sign factors}

In this section, we show that the sign factors play a role to divide the domain of the delta functions on the real twistor space by $O(1)$, for the first time.

\begin{proposition}
The domain of $\tilde{\delta}^{(3)}_{-n-4}(W,Y)$ defined by Eq. $(\ref{tdelta})$ without the sign factor is $\mathbb{RP}^{3} \times O(1) \approx S^{3}$.
However, the domain of ${\delta}^{(3)}_{-n-4}(W,Y)$ defined by Eq. $(\ref{d1})$ with the sign factor is $\mathbb{RP}^{3}$. 
\end{proposition}

\begin{proof}
First, we carry out the integration of Eq. $(\ref{d1})$
in subset 
$U = \{ W=(W_{0},W_{1},W_{2},W_{3}) \in M_{1,4}(\mathbb{R}) | W_{0} \ne 0  \}$, where
$M_{1,4}(\mathbb{R})$ is the set of all nonzero vectors in $\mathbb{R}^{4}$.
By using $W_{0} \ne 0$ and $Y_{0} \ne 0$ in $U$, we obtain 
\begin{align}
\delta_{-n-4}^{(3)} (W,Y)
&= \int_{-\infty}^{\infty} \dfrac{dt}{t^{1+n}} \text{sgn} (t)
\delta (W_{0}-tY_{0}) \prod_{\alpha=1}^{3} \delta (W_{\alpha}-tY_{\alpha})
\notag \\
&= \left( \dfrac{Y_{0}}{W_{0}} \right)^{1+n} \text{sgn} \left( \dfrac{W_{0}}{Y_{0}} \right) \dfrac{1}{|Y_{0}|}
\prod_{\alpha=1}^{3} \delta \left(W_{\alpha}-\dfrac{W_{0}}{Y_{0}}Y_{\alpha} \right)
\notag \\
&= \left( \dfrac{Y_{0}}{W_{0}} \right)^{n} \dfrac{1}{|W_{0}|^{4}} 
\prod_{\alpha=1}^{3} \delta \left(\dfrac{W_{\alpha}}{W_{0}}-\dfrac{Y_{\alpha}}{Y_{0}} \right).
\label{drp3}
\end{align}
Next, we carry out the integration of Eq. $(\ref{tdelta})$
in $U$, such that 
\begin{align}
\tilde{\delta}_{-n-4}^{(3)}(W,Y) 
&:= \int_{-\infty}^{\infty} \dfrac{dt}{t^{1+n}} \delta \left( W_{0} - t Y_{0} \right)
\prod_{\alpha=1}^{3} \delta \left( W_{\alpha} - t Y_{\alpha} \right)
\notag \\
&= \int_{-\infty}^{\infty} \dfrac{dt}{t^{1+n}} \dfrac{1}{|Y_{0}|} \delta \left( \dfrac{W_{0}}{Y_{0}} - t \right)
\prod_{\alpha=1}^{3} \delta \left( W_{\alpha} - t Y_{\alpha} \right)
\notag \\
&= \left( \dfrac{Y_{0}}{W_{0}} \right)^{1+n} \dfrac{1}{|Y_{0}||W_{0}|^{3}} \prod_{\alpha=1}^{3} \delta 
\left( \dfrac{W_{\alpha}}{W_{0}} - \dfrac{Y_{\alpha}}{Y_{0}} \right).
\label{ds3}
\end{align}
From Eqs. $(\ref{drp3})$ and $(\ref{ds3})$, we see the relationship
\begin{align}
\delta^{(3)}_{-n-4}(W,Y)
= \dfrac{Y_{0}}{|Y_{0}|}
\dfrac{W_{0}}{|W_{0}|} \,\tilde{\delta}^{(3)}_{-n-4}(W,Y).
\label{tdfadf}
\end{align}
Therefore, we obtain 
\begin{align}
\int_{U} f(W) \tilde{\delta}^{(3)}_{-n-4}(W,Y) D^{3}W
= \dfrac{Y_{0}}{|Y_{0}|} \dfrac{W_{0}}{|W_{0}|} f(Y)
\label{idfd}
\end{align}
where $f(W)$ is a homogeneous function of degree $n$, by performing integration in domain $U$. 

Here, $\mathbb{RP}^{3}$ is the set of all lines in $\mathbb{R}^{4}$ through the origin.  
This is the quotient set of
 $M_{1,4}(\mathbb{R})$ 
by the equivalence relation in which the elements multiplied by a nonzero real number are equivalent:  
\begin{align}
\mathbb{RP}^{3} = M_{1,4}(\mathbb{R})/\mathbb{R}^{*}
\end{align}
where $\mathbb{R}^{*}:=\mathbb{R}-\{0\} \cong O(1) \times \mathbb{R}^{+}$. 
Here, $\mathbb{R}^{+}:=\{x \in \mathbb{R} \,|\, x>0\}$.
In $U$,
for $W \in M_{1,4}(\mathbb{R})$ and $g \in \mathbb{R}^{+}$, the sign of $W_{0}$ is the same as that of $(gW)_{0}$. Therefore, in $M_{1,4}(\mathbb{R})/\mathbb{R}^{+}$, the elements with $W_{0}>0$ differ from those with $W_{0}<0$.
By identifying these elements, that is, 
dividing $M_{1,4}(\mathbb{R})/\mathbb{R}^{+}$ by $O(1)$, we obtain $\mathbb{RP}^{3}$. Therefore, $M_{1,4}(\mathbb{R})/\mathbb{R}^{+} \approx \mathbb{RP}^{3} \times O(1)$.
By comparing this fact with  
\begin{align}
\mathbb{RP}^{3} \approx O(4)/[O(1) \times O(3)], 
\end{align}
we can see 
\begin{align}
M_{1,4}(\mathbb{R})/\mathbb{R}^{+} \approx
O(4)/O(3) \approx S^{3}.
\end{align}
From Eqs. $(\ref{tdfadf})$ and $(\ref{idfd})$, we observe that the tilded $\delta$-function $\tilde{\delta}^{(3)}_{-n-4}(W,Y)$ is equivalent to the delta function $\delta^{(3)}_{-n-4}(W,Y)$ in the subset of $U$ which is imposed the condition 
$W_{0}>0$ ($Y_{0}>0$) or $W_{0}<0$ ($Y_{0}<0$). 
Therefore, $\tilde{\delta}^{(3)}_{-n-4}(W,Y)$ is equal to 
${\delta}^{(3)}_{-n-4}(W,Y)$ in $M_{1,4}(\mathbb{R})/\mathbb{R}^{+} \approx \mathbb{RP}^{3} \times O(1) \approx S^{3}$. Hence, the domain of $\tilde{\delta}^{(3)}_{-n-4}(W,Y)$ is $S^{3}$. 
\end{proof}

\begin{proposition}
The domain of ${\delta}^{(2)}(W_{1},W_{2},W_{3})$ defined by  Eq. $(\ref{w123})$ with the sign factor is $G_{2,4}(\mathbb{R})$.
However,
the domain of $\tilde{\delta}^{(2)}(W_{1},W_{2},W_{3})$ defined by Eq. $(\ref{delta2})$ without the sign factor is $G_{2,4}(\mathbb{R}) \times O(1)$. 
\end{proposition}

\begin{proof}
We consider Eqs. $(\ref{w123})$ and $(\ref{delta2})$
in a subset
\begin{align}
V = \left\{\,
\begin{pmatrix}
B_{11} & B_{12} & B_{13} & B_{14} \\
B_{21} & B_{22} & B_{23} & B_{24}
\end{pmatrix}
\in M_{2,4}(\mathbb{R})
\, \middle| \,
\text{det} 
\begin{pmatrix}
B_{11} & B_{12} \\
B_{21} & B_{22}
\end{pmatrix}
\ne 0
\,\right\},
\label{V}
\end{align}
where $M_{2,4}(\mathbb{R})$ is the set of all 2 $\times$ 4 matrices with a rank of 2. First, we compose a 2 $\times$ 4 matrix by vertically arranging $W_{2}=( \lambda_{2A}, \mu_{2}^{A^{\prime}} )$ and $W_{3}=( \lambda_{3A}, \mu_{3}^{A^{\prime}} )$.  This is an element of $M_{2,4}(\mathbb{R})$ because $W_{2}$ and $W_{3}$ are assumed to be linearly independent.
For this matrix, we have 
 \begin{align}
\text{det} 
\begin{pmatrix}
\lambda_{20} & \lambda_{21} \\
\lambda_{30} & \lambda_{31}
\end{pmatrix}
= \lambda_{20} \lambda_{31} - \lambda_{21} \lambda_{30}
= \langle 2 3 \rangle 
= - \langle 3 2 \rangle
\ne 0
\end{align}
in $V$, that is, $\lambda_{2}$ and $\lambda_{3}$ are linearly independent.
Therefore, we can express the delta function of $\lambda$ on the right-hand side of Eq. $(\ref{w123})$ and Eq. $(\ref{delta2})$ as follows:
\begin{align}
\delta^{(2)} 
(\lambda_{1} - s \lambda_{2} - t \lambda_{3})
= \dfrac{1}{|\langle 3 2 \rangle|}
\delta \left( s-\dfrac{\langle 31 \rangle}{\langle 32 \rangle} \right)
\delta \left( t-\dfrac{\langle 12 \rangle}{\langle 32 \rangle} \right).
\end{align}
Here, we have $\langle 31 \rangle \ne 0$ and $\langle 12 \rangle \ne 0$, because the 2 $\times$ 4 matrix composed by vertically arranging $W_{3}=( \lambda_{3A}, \mu_{3}^{A^{\prime}} )$ and $W_{1}=( \lambda_{1A}, \mu_{1}^{A^{\prime}} )$ belongs to $V$ and that composed by vertically arranging $W_{1}=( \lambda_{1A}, \mu_{1}^{A^{\prime}} )$ and $W_{2}=( \lambda_{2A}, \mu_{2}^{A^{\prime}} )$ also belongs to $V$.
Therefore, by carring out the integration in 
 Eqs. $(\ref{w123})$ and $(\ref{delta2})$, we have
\begin{align}
\delta^{(2)} (W_{1},W_{2},W_{3})
&= \left| \dfrac{\langle 32 \rangle}{\langle 31 \rangle \langle 12 \rangle} \right|
\delta^{(2)} \left( \mu_{1} - \dfrac{\langle 31 \rangle}{\langle 32 \rangle} \mu_{2} - \dfrac{\langle 12 \rangle}{\langle 32 \rangle} \mu_{3}\right),
\\
\tilde{\delta}^{(2)} (W_{1},W_{2},W_{3})
&= \dfrac{|\langle 32 \rangle|}{\langle 31 \rangle \langle 12 \rangle} 
\delta^{(2)} \left( \mu_{1} - \dfrac{\langle 31 \rangle}{\langle 32 \rangle} \mu_{2} - \dfrac{\langle 12 \rangle}{\langle 32 \rangle} \mu_{3}\right).
\end{align}
From these equations, we can confirm the relationship in Eq. $(\ref{3112})$:
\begin{align}
\delta^{(2)} (W_{1},W_{2},W_{3})
= \dfrac{\langle 31 \rangle \langle 12 \rangle}{|\langle 31 \rangle \langle 12 \rangle|}\,
\tilde{\delta}^{(2)} (W_{1},W_{2},W_{3}).
\label{W123}
\end{align}

The Grassmann manifold $G_{2,4}(\mathbb{R})$ is the set of all two-dimensional planes in $\mathbb{R}^{4}$.
This is equivalent to the quotient space $M_{2,4}(\mathbb{R})/GL(2,\mathbb{R})$. 
Here, we have $GL(2,\mathbb{R}) \approx O(1) \times \mathbb{R}^{+} \times SL(2,\mathbb{R})$ because of $GL(2,\mathbb{R})/SL(2,\mathbb{R})
\cong \mathbb{R}^{*} \cong O(1) \times \mathbb{R}^{+}$. Therefore, 
\begin{align}
G_{2,4}(\mathbb{R}) \approx 
\dfrac{M_{2,4}(\mathbb{R})}{O(1) \times \mathbb{R}^{+} \times SL(2,\mathbb{R})}.
\label{g24v1}
\end{align}
Furthermore, for $g \in \mathbb{R}^{+} \times SL(2,\mathbb{R})$, the sign of the determinant
of $\begin{pmatrix} B_{11} & B_{12} \\ B_{21} & B_{22} \end{pmatrix}$ of $V$
is the same as that of $g\begin{pmatrix} B_{11} & B_{12} \\ B_{21} & B_{22} \end{pmatrix}$
because of $\det g >0$. 
Therefore, in $M_{2,4}(\mathbb{R})/[\mathbb{R}^{+} \times SL(2,\mathbb{R})]$, the elements with $\det \begin{pmatrix} B_{11} & B_{12} \\ B_{21} & B_{22} \end{pmatrix}>0$ differ from those with
$\det \begin{pmatrix} B_{11} & B_{12} \\ B_{21} & B_{22} \end{pmatrix}<0$. 
By identifying these elements, that is, dividing $M_{2,4}(\mathbb{R})/[\mathbb{R}^{+} \times SL(2,\mathbb{R})]$ by $O(1)$, we obtain $G_{2,4}(\mathbb{R})$. Therefore,
${M_{2,4}(\mathbb{R})}/[{\mathbb{R}^{+} \times SL(2,\mathbb{R})}] \approx G_{2,4}(\mathbb{R}) \times O(1)$.
Here, from Eq. $(\ref{W123})$, we see that $\tilde{\delta}^{(2)}(W_{1},W_{2},W_{3})$ is equal to 
$\delta^{(2)}(W_{1},W_{2},W_{3})$
in the subset of $V$ which is imposed the condition  
$\det \begin{pmatrix} B_{11} & B_{12} \\ B_{21} & B_{22} \end{pmatrix}>0$
(i.e. $\langle 1 2 \rangle>0$, $\langle 3 1 \rangle>0$ in this subset) or the condition $\det \begin{pmatrix} B_{11} & B_{12} \\ B_{21} & B_{22} \end{pmatrix}<0$
(i.e. $\langle 1 2 \rangle<0$, $\langle 3 1 \rangle<0$ in this subset). 
Therefore, $\tilde{\delta}^{(2)}(W_{1},W_{2},W_{3})$ is equal to $\delta^{(2)}(W_{1},W_{2},W_{3})$ in 
${M_{2,4}(\mathbb{R})}/[{\mathbb{R}^{+} \times SL(2,\mathbb{R})}] \approx G_{2,4}(\mathbb{R}) \times O(1)$.
Hence, the domain of $\tilde{\delta}^{(2)}(W_{1},W_{2},W_{3})$ is $G_{2,4}(\mathbb{R})\times O(1)$ and that of $\delta^{(2)}(W_{1},W_{2},W_{3})$
is $G_{2,4}(\mathbb{R})$.
\end{proof}

\section{New definition of the delta functions on the complex twisor space}

\subsection{New definition of the delta function 
on space of complex numbers in terms of the \v{C}ech cohomology group}

To define a new delta function on 
$3$-dimensional complex projective space $\mathbb{CP}^{3}$ in terms of the \v{C}ech cohomology group, 
we propose a new definition of the delta function on $n$-dimensional complex space $\mathbb{C}^{n}$ in terms of the \v{C}ech cohomology group in this subsection.

\begin{definition}
Let $(z_{1},z_{2},\cdot,z_{n})$ be the coordinates of point $z$ in $n$-dimensional complex space $\mathbb{C}^{n}$.
We define the delta function on $\mathbb{C}^{n}$ as 
\begin{align}
\delta^{(n)}_{\mathbb{C}}(z_{1},z_{2},\cdots,z_{n}) :=
\dfrac{1}{(2 \pi i)^{n}} \dfrac{1}{z_{1} z_{2} \cdots z_{n}}.
\label{Cdelta}
\end{align}
\end{definition}

This resembles the complex function which represents the Sato hyperfunction \cite{Ka}.

\begin{proposition}
For a complex function $f(z_{1},z_{2},\cdots, z_{n})$ which is holomorphic in a suitable region in $\mathbb{C}^{n}$, we have   
\begin{align}
&\int_{\Gamma} \delta^{(n)}_{\mathbb{C}}(z_{1},z_{2},\cdots,z_{n}) f(z_{1}, z_{2}, \cdots, z_{n})
dz_{1} dz_{2} \cdots dz_{n}
\notag \\
=
& \int_{C_{1}} \int_{C_{2}} \cdots \int_{C_{n}}
\dfrac{1}{(2 \pi i)^{n}} \dfrac{1}{z_{1} z_{2} \cdots z_{n}}
f(z_{1}, z_{2}, \cdots, z_{n})
dz_{1} dz_{2} \cdots dz_{n}
\notag \\
=
&f(0,0,\cdots,0)
\label{fdelta}
\end{align}
using Cauchy's integration formula, where the integral region $\Gamma$ is the $n$-dimensional torus
\begin{align}
\Gamma = C_{1} \times C_{2} \times \cdots \times C_{n},
\,\,\,\,\,
C_{k} = \left\{ z_{k} = r_{k} e^{i\theta_{k}} \,\middle|\,
0 \le \theta_{k} \le 2 \pi \right\}, \,\,\,\,\, (k=1,2,\cdots,n). 
\label{Gamma}
\end{align}
\label{propo4}
\end{proposition}

\begin{remark}
Proposition $\ref{propo4}$ holds for the delta function $\delta^{(n)}_{\mathbb{C}}(z_{1},z_{2},\cdots,z_{n})$ with the extra terms
\begin{align}
\sum_{
\begin{matrix}
m_{1},m_{2},\cdots,m_{n} \ge 0
\\
-m_{1}-m_{2}-\cdots -m_{n-1}+m_{n}=-n
\end{matrix}
}
\Bigg(
&A_{m_{1} m_{2} \cdots m_{m}}
\dfrac{(z_{n})^{m_{n}}}{(z_{1})^{m_{1}} (z_{2})^{m_{2}} \cdots
(z_{n-1})^{m_{n-1}}}
\notag \\
&+
B_{m_{1} m_{2} \cdots m_{m}}
\dfrac{(z_{n-1})^{m_{n}}}{(z_{1})^{m_{1}} \cdots (z_{n-2})^{m_{n-2}}
(z_{n})^{m_{n-1}}}
\notag \\
&
+ \cdots
\notag \\
&+
C_{m_{1} m_{2} \cdots m_{m}}
\dfrac{(z_{1})^{m_{n}}}{(z_{2})^{m_{1}} (z_{3})^{m_{2}} \cdots
(z_{n})^{m_{n-1}}}
\Bigg)
\end{align}
in right hand of Eq. $(\ref{Cdelta})$.
Here,  $A_{m_{1} m_{2} \cdots m_{m}}$, $B_{m_{1} m_{2} \cdots m_{m}}$, and $C_{m_{1} m_{2} \cdots m_{m}}$
are arbitrary complex constants.
Therefore, the delta function $\delta^{(n)}_{\mathbb{C}}(z_{1},z_{2},\cdots,z_{n})$ is 
a representative element of  the $(n-1)$-th \v{C}ech cohomology group on $\mathbb{CP}^{n-1}$ with coefficients in the sheaf $\mathcal{O}(-n)$, 
denoted by $H^{n-1}(\mathbb{CP}^{n-1},\mathcal{O}(-n))$, where $\mathbb{CP}^{n-1}$ is covered by
\begin{align}
U_{i} = \left\{ (z_{1},z_{2},\cdots,z_{n}) \in \mathbb{CP}^{n-1}
\,\middle|\,
z_{i} \ne 0
\right\}, \,\,\,\,\, i=1,2,\cdots,n.
\end{align}  
\label{remark1}
\end{remark}

\begin{theorem}
For an arbitrary $i=1,2,\cdots,n$, we have
\begin{align}
z_{i} \,\delta^{(n)}_{\mathbb{C}}(z_{1},z_{2},\cdots,z_{n})=0.
\end{align}
\end{theorem}

\begin{proof}
For a complex function $f(z_{1},z_{2},\cdots, z_{n})$ which is holomorphic in a suitable region in $\mathbb{C}^{n}$, we have
\begin{align}
&\int_{\Gamma} z_{i}\,\delta^{(n)}_{\mathbb{C}}(z_{1},z_{2},\cdots,z_{z}) f(z_{1}, z_{2}, \cdots, z_{n})
dz_{1} dz_{2} \cdots dz_{n}
\\
=
& \int_{C_{1}} \int_{C_{2}} \cdots \int_{C_{n}}
\dfrac{1}{(2 \pi i)^{n}} \dfrac{1}{z_{1} \cdots z_{i-1} z_{i+1} \cdots z_{n}}
f(z_{1}, z_{2}, \cdots, z_{n})
dz_{1} dz_{2} \cdots dz_{n}
\\
=
&0
\label{fdelta}
\end{align}
using Cauchy's integration formula, particularly for variable $z_{i}$. The integral region $\Gamma$ is defined by Eq. $(\ref{Gamma})$.
\end{proof}

\begin{theorem}
For a complex non-singular matrix $A=(A_{ij})$ ( $i,j=1,2,\cdots n$ ) and an $n$-dimensional complex vector $\mathbf{z}=(z_{1}, \cdots, z_{n})$, we have
\begin{align}
\delta_{\mathbb{C}}^{(n)} \left( A \mathbf{z} \right)
= \dfrac{1}{\det A} \delta_{\mathbb{C}}^{(n)} \left( \mathbf{z} \right).
\label{Az}
\end{align}
\label{thAz}
\end{theorem}

\begin{proof}
From Eq. $(\ref{Cdelta})$, $\delta_{\mathbb{C}}^{(n)} \left( A \mathbf{z} \right)$ is represented as follow:
\begin{align}
\delta_{\mathbb{C}}^{(n)} \left( A \mathbf{z} \right)
= \dfrac{1}{(2 \pi i)^{n}}
\dfrac{1}{w_{1} \cdots w_{n}}, \quad
w_{i} = \sum_{j=1}^{n}A_{ij} z_{j}.
\end{align}
Here, we have $z_{i}=\sum_{j=1}^{n}A^{-1}_{ij} w_{j}$ ; therefore, 
\begin{align}
d^{n} \mathbf{z}
= dz_{1} \wedge dz_{2} \wedge \cdots \wedge dz_{n}
= (\det A^{-1}) dw_{1} \wedge dw_{2} \wedge \cdots
\wedge dw_{n}
= \dfrac{1}{\det A}
d^{n} \mathbf{w}. 
\end{align}
Hence, for a complex function $f(\mathbf{z})$ which is holomorphic in a suitable region in $\mathbb{C}^{n}$, we can see that
\begin{align}
& \quad
\int_{\Gamma} \delta^{(n)}_{\mathbb{C}} (A \mathbf{z})
f(\mathbf{z}) 
d^{n} \mathbf{z}
\notag \\
&= \int_{\Gamma} \dfrac{1}{(2 \pi i)^{n}}
\dfrac{1}{w_{1} \cdots w_{n}} f(A^{-1}_{ij} w_{j}) 
\dfrac{1}{\det A} d^{n} \mathbf{w}
\notag \\
&= \dfrac{1}{\det A} f(0,\cdots,0)
\notag \\
& = \dfrac{1}{\det A} \int_{\Gamma} \delta^{(n)}_{\mathbb{C}} (\mathbf{z}) f(\mathbf{z}) d^{n} \mathbf{z}  
\end{align}
using Cauchy's integration formula. 
Here, the integral region $\Gamma$ is defined by Eq. $(\ref{Gamma})$.
Thus, we can see $\delta^{(n)}_{\mathbb{C}}(A \mathbf{z})=\delta^{(n)}_{\mathbb{C}}(\mathbf{z})/\det A$.
\end{proof}

\begin{remark}
We regard the $3$-dimensional complex projective space
$\mathbb{CP}^{3}$ as the dual complex twistor space $\mathbb{PT}^{*}$ to investigate the relationship between 
the delta functions and the scattering amplitudes in $\mathcal{N}=4$ SYM in the followings. 
\end{remark}

Here, we propose a new definition of the delta function on $\mathbb{PT}^{*}$ using the delta function on $\mathbb{C}^{4}$. This is constructed by considering the equivalence relation of $\mathbb{PT}^{*}$ where two points $\bar{Z}=\left(\bar{Z}_{\alpha}\right)$ and $\bar{X}=\left(\bar{X}_{\alpha}\right)$ are identified when a nonzero complex number $w$ exists, such that $\bar{Z}_{\alpha} = w \bar{X}_{\alpha}$.

\begin{definition}
We define the delta function on the dual complex twistor space $\mathbb{PT}^{*}$ as follows:
\begin{align}
\Delta_{m}^{(3)}(\bar{Z},\bar{X}) 
:=
\int_{C} \delta^{(4)}_{\mathbb{C}} \left(\bar{Z}_{\alpha}-w\bar{X}_{\alpha}\right) 
w^{m+3} dw,
\label{delta}
\end{align}
where the integral contour $C$ surrounds only one of the four singular points of 
$\delta^{(4)}_{\mathbb{C}} \left(\bar{Z}_{\alpha}-w\bar{X}_{\alpha}\right)$.
\end{definition}

When we consider the covering
\begin{align}
U_{0} = \left\{ \left(\bar{X}_{0},\bar{X}_{1},\bar{X}_{2},\bar{X}_{3}\right) \in \mathbb{CP}^{3} \,\middle| \,\bar{X}_{0} \ne 0 \right\},
\label{U0}
\end{align}
the integral contour $C$ is selected to surround only the singular point of $1/(\bar{Z}_{0}-w\bar{X}_{0})$, that is, $w=\bar{Z}_{0}/\bar{X}_{0}$. Hence, we obtain
\begin{align}
\Delta_{m}^{(3)}(\bar{Z},\bar{X})
= \dfrac{1}{(2 \pi i)^{3}} 
\dfrac{\left(\bar{Z}_{0}\right)^{m+3}}{\left(\bar{X}_{0}\right)^{m+1} \left(\bar{X}_{1}\bar{Z}_{0}-\bar{X}_{0}\bar{Z}_{1}\right)\left(\bar{X}_{2}\bar{Z}_{0}-\bar{X}_{0}\bar{Z}_{2}\right)\left(\bar{X}_{3}\bar{Z}_{0}-\bar{X}_{0}\bar{Z}_{3}\right)},
\label{Deltam}
\end{align}
by carrying out the integration in Eq. $(\ref{delta})$.
Here, we denote the inhomogeneous coordinates of $U_{0}$ by
\begin{align}
\zeta_{1} = \dfrac{\bar{X}_{1}}{\bar{X}_{0}}, \quad
\zeta_{2} = \dfrac{\bar{X}_{2}}{\bar{X}_{0}}, \quad
\zeta_{3} = \dfrac{\bar{X}_{3}}{\bar{X}_{0}}, 
\label{zeta}
\end{align}
\begin{align}
\eta_{1} = \dfrac{\bar{Z}_{1}}{\bar{Z}_{0}}, \quad
\eta_{2} = \dfrac{\bar{Z}_{2}}{\bar{Z}_{0}}, \quad
\eta_{3} = \dfrac{\bar{Z}_{3}}{\bar{Z}_{0}}.
\label{eta}
\end{align}
Using these coordinates, Eq. $(\ref{Deltam})$ can be represented as 
\begin{align}
\Delta_{m}^{(3)}(\bar{Z},\bar{X})
&= \dfrac{1}{(2 \pi i)^{3}} 
\dfrac{\left(\bar{Z}_{0}\right)^{m}}{\left(\bar{X}_{0}\right)^{m+4}}
\dfrac{1}{(\zeta_{1}-\eta_{1})(\zeta_{2}-\eta_{2})(\zeta_{3}-\eta_{3})}
\notag
\\
&= \dfrac{\left(\bar{Z}_{0}\right)^{m}}{\left(\bar{X}_{0}\right)^{m+4}}
\delta^{(3)}_{\mathbb{C}}
(\zeta_{1}-\eta_{1},\zeta_{2}-\eta_{2},\zeta_{3}-\eta_{3}).
\label{3Delta}
\end{align}
From this result and Remark $\ref{remark1}$, we see that $\Delta_{m}^{(3)}(\bar{Z},\bar{X})$ is a representative element of $H^{2}(\mathbb{CP}^{2},\mathcal{O}(-3))$.

\begin{proposition}
$\Delta_{m}^{(3)}(\bar{Z},\bar{X})$ has a degree of homogeneity $m$ in $\bar{Z}$
and $(-m-4)$ in $\bar{X}$.  
\end{proposition}

\begin{proof}
From Eq. $(\ref{3Delta})$, for non-zero complex numbers $a$ and $b$, we have
\begin{align}
\Delta_{m}^{(3)}\left(a\bar{Z},\bar{X}\right) = a^{m} \Delta_{m}^{(3)}(\bar{Z},\bar{X}),
\quad
\Delta_{m}^{(3)}\left(\bar{Z},b\bar{X}\right) = \dfrac{1}{b^{m+4}}
\Delta_{m}^{(3)}(\bar{Z},\bar{X}).
\notag
\end{align}
\end{proof}

\begin{theorem}
For a homogeneous holomorphic function 
$f_{m}\left(\bar{X}_{0},\bar{X}_{1},\bar{X}_{2},\bar{X}_{3}\right)$
on $\mathbb{PT}^{*}$ of degree $m$, we have
\begin{align}
\int_{\Gamma^{\prime}} f_{m} \left(\bar{X}_{0},\bar{X}_{1},\bar{X}_{2},\bar{X}_{3}\right) 
\Delta_{m}^{(3)}(\bar{Z},\bar{X}) D^{3} \bar{X} = f_{m} \left(\bar{Z}_{0},\bar{Z}_{1},\bar{Z}_{2},\bar{Z}_{3}\right),
\end{align}
where the integral contour $\Gamma^{\prime}$ surrounds all singular points of the delta function $\Delta_{m}^{(3)}(\bar{Z},\bar{X})$ and the volume element of $\mathbb{PT}^{*}$ is defined by $D^{3}\bar{X}:=\dfrac{1}{6} \epsilon^{\alpha \beta \gamma \delta} \bar{X}_{\alpha} d\bar{X}_{\beta} \wedge d\bar{X}_{\gamma} \wedge d\bar{X}_{\delta}$. Similarly, 
for a homogeneous holomorphic function 
$f_{-m-4} \left(\bar{X}_{0},\bar{X}_{1},\bar{X}_{2},\bar{X}_{3}\right)$
on $\mathbb{PT}^{*}$ of degree $(-m-4)$, we have 
\begin{align}
\int_{\Gamma^{\prime}} f_{-m-4} \left(\bar{Z}_{0},\bar{Z}_{1},\bar{Z}_{2},\bar{Z}_{3}\right) 
\Delta_{m}^{(3)}(\bar{Z},\bar{X}) D^{3} \bar{Z} = f_{-m-4} \left(\bar{X}_{0},\bar{X}_{1},\bar{X}_{2},\bar{X}_{3}\right).
\label{fzdelta}
\end{align}
\end{theorem}

\begin{proof}
By using Eq. $(\ref{zeta})$, we can express 
$
D^{3}\bar{X} = \left(\bar{X}_{0}\right)^{4} d \zeta_{1} \wedge
d \zeta_{2} \wedge d \zeta_{3}
$
and $f_{m}\left(\bar{X}_{0},\bar{X}_{1},\bar{X}_{2},\bar{X}_{3}\right)=\left(\bar{X}_{0}\right)^{m}f_{m}\left(1,\zeta_{1},\zeta_{2},\zeta_{3}\right)$. With these expressions and
Eq. $(\ref{3Delta})$, we obtain  
\begin{align}
&
\int_{\Gamma^{\prime}} f_{m} \left(\bar{X}_{0},\bar{X}_{1},\bar{X}_{2},\bar{X}_{3}\right) 
\Delta_{m}^{(3)}(\bar{Z},\bar{X}) D^{3} \bar{X}
\notag
\\
&= \left(\bar{Z}_{0} \right)^{m} 
\int_{\Gamma^{\prime}} 
f_{m}\left(1,\zeta_{1},\zeta_{2},\zeta_{3}\right)
\dfrac{1}{(2\pi i)^{3}} \dfrac{1}{(\zeta_{1}-\eta_{1})(\zeta_{2}-\eta_{2})(\zeta_{3}-\eta_{3})}
d\zeta_{1} \wedge d\zeta_{2} \wedge d\zeta_{3}
\notag
\\
&= \left(\bar{Z}_{0}\right)^{m} f_{m}\left(1,\eta_{1},\eta_{2},\eta_{3}\right)
\notag
\\
&= f_{m} \left(\bar{Z}_{0},\bar{Z}_{1},\bar{Z}_{2},\bar{Z}_{3}\right). 
\end{align}
Similarly, we can obtain Eq. $(\ref{fzdelta})$. 
\end{proof}

\begin{definition}
Next, we define the delta function which enforces the  collinearity of the three twistors $\bar{Z}_{1} =\left(\bar{Z}_{1 \alpha} \right)$, $\bar{Z}_{2} =\left(\bar{Z}_{2 \alpha} \right)$, and
$\bar{Z}_{3} =\left(\bar{Z}_{3 \alpha} \right)$ in $\mathbb{PT}^{*}$ as follows: 
\begin{align}
\,\,\,\,\,
\Delta^{(2)}\left(\bar{Z}_{1},\bar{Z}_{2},\bar{Z}_{3}\right) := \int_{\gamma} 
\dfrac{dv}{v} \dfrac{dw}{w} \delta^{(4)}_{\mathbb{C}}
\left(\bar{Z}_{1 \alpha}-v\bar{Z}_{2 \alpha}-w\bar{Z}_{3 \alpha}\right),
\label{2Delta}
\end{align}
where the integral contour $\gamma$ surrounds the point $(v,w)$ in $\mathbb{C}^{2}$ which satisfies
\begin{align}
\bar{Z}_{10} - v \bar{Z}_{20} - w \bar{Z}_{30} = 0,\,\,\,\,\,
\bar{Z}_{11} - v \bar{Z}_{21} - w \bar{Z}_{31} = 0.
\label{gamma}
\end{align}
\end{definition}

\begin{remark}
We use the notation $\bar{Z}_{i}=\left(\bar{Z}_{i \alpha}\right)
=\left(\tilde{\pi}_{iA},\tilde{\omega}_{i}^{A^{\prime}}\right)$
( $A=0,1;A^{\prime}=0^{\prime},1^{\prime}$ ) 
and redefine
$\langle i j \rangle 
:= \tilde{\pi}_{iA} \tilde{\pi}_{j}^{A} 
= \tilde{\pi}_{iA} \epsilon^{AB} \tilde{\pi}_{jB} 
= \tilde{\pi}_{i0} \tilde{\pi}_{j1} - \tilde{\pi}_{i1} \tilde{\pi}_{j0}$
in the followings.
\end{remark}

\begin{proposition}
By carrying out the integration in Eq. $(\ref{2Delta})$, we have
\begin{align}
\Delta^{(2)}\left(\bar{Z}_{1},\bar{Z}_{2},\bar{Z}_{3}\right) 
= 
\dfrac{\langle 2 3 \rangle^{3}}
{\langle 3 1 \rangle 
\langle 1 2 \rangle}
\delta^{(2)}_{\mathbb{C}}
\left(
\langle 2 3 \rangle 
\tilde{\omega}_{1}^{A^{\prime}}
+ \langle 3 1 \rangle 
\tilde{\omega}_{2}^{A^{\prime}}
+ \langle 1 2 \rangle 
\tilde{\omega}_{3}^{A^{\prime}}
\right).
\end{align}
This is a representative element of $H^{1}(\mathbb{CP}^{1}, \mathcal{O}(-2))$.
\end{proposition}

\begin{proof}
First, we transform the integral variables $v$ and $w$
into $a$ and $b$, respectively, as follows:
\begin{align}
a= v\bar{Z}_{20} + w\bar{Z}_{30}, \,\,\,\,\, 
b= v\bar{Z}_{21} + w\bar{Z}_{31}.
\label{ab}
\end{align}
Using the variables $a$ and $b$, the integral contour $\gamma$ in Eq. $(\ref{2Delta})$ is represented as a $2$-dimensional torus which is the direct product of the circle $\gamma_{0}$ with center $a=\bar{Z}_{10}$ and the circle $\gamma_{1}$ with center $b=\bar{Z}_{11}$: $\gamma = \gamma_{0} \times \gamma_{1}$, where
\begin{align}
&\gamma_{0} = 
\left\{ a=\bar{Z}_{10} + \varepsilon_{0} e^{i \theta_{0}}
\,\middle|\, 
\varepsilon_{0}>0,\, 
0 \le \theta_{0} \le 2\pi \right\},
\label{gamma0}
\\
&\gamma_{1} = 
\left\{ b=\bar{Z}_{11} + \varepsilon_{1} e^{i \theta_{1}}
\,\middle|\, 
\varepsilon_{1}>0,\, 
0 \le \theta_{1} \le 2\pi \right\}.
\label{gamma1}
\end{align}
Here, from Eq. $(\ref{ab})$, we obtain
\begin{align}
v = \dfrac{a\bar{Z}_{31}-b\bar{Z}_{30}}{\langle 2 3 \rangle}, \,\,\,\,\,
w = \dfrac{-a\bar{Z}_{21}+b\bar{Z}_{20}}{\langle 2 3 \rangle},
\,\,\,\,\,
dv \wedge dw = \dfrac{1}{\langle 2 3 \rangle} da \wedge db.
\label{vwab}
\end{align}
Using these relations, we can carry out the integration in Eq. $(\ref{2Delta})$ in terms of variables $a$ and $b$ as follows:
\begin{align}
&\Delta^{(2)}\left(\bar{Z}_{1},\bar{Z}_{2},\bar{Z}_{3}\right)
\notag
\\
&=
\dfrac{\langle 2 3 \rangle}{(2 \pi i)^{2}}
\int_{\gamma} da \wedge db
\dfrac{1}{\left(a-\bar{Z}_{10}\right)\left(b-\bar{Z}_{11}\right)}
\dfrac{1}{\left(a\bar{Z}_{31}-b\bar{Z}_{30}\right)\left(-a\bar{Z}_{21}+b\bar{Z}_{20}\right)}
\notag
\\
&
\,\,\,\,\,
\times
\prod_{{\alpha}=2}^{3}
\dfrac{1}{(2\pi i)}
\dfrac{1}{\bar{Z}_{1\alpha}
-\dfrac{a\bar{Z}_{31}-b\bar{Z}_{30}}{\langle 2 3 \rangle}\bar{Z}_{2\alpha}
-\dfrac{-a\bar{Z}_{21}+b\bar{Z}_{20}}{\langle 2 3 \rangle}\bar{Z}_{3\alpha}}
\\
&=\dfrac{\langle 2 3 \rangle}
{\langle 3 1 \rangle \langle 1 2 \rangle}
\prod_{\alpha=2}^{3}
\dfrac{1}{(2\pi i)}
\dfrac{\langle 2 3 \rangle}
{\langle 2 3 \rangle \bar{Z}_{1 \alpha}
+\langle 3 1 \rangle \bar{Z}_{2 \alpha}
+\langle 1 2 \rangle \bar{Z}_{3 \alpha}}
\\
&
=\dfrac{\langle 2 3 \rangle^{3}}{\langle 3 1 \rangle \langle 1 2 \rangle}
\delta^{(2)}_{\mathbb{C}}
\left(\langle 2 3 \rangle \tilde{\omega}_{1}^{A^{\prime}} 
+ \langle 3 1 \rangle \tilde{\omega}_{2}^{A^{\prime}} 
+ \langle 1 2 \rangle \tilde{\omega}_{3}^{A^{\prime}}
\right).
\end{align}
From this result and Remark $\ref{remark1}$, we can see that $\Delta^{(2)}\left(\bar{Z}_{1},\bar{Z}_{2},\bar{Z}_{3}\right)$ is a representative element of $H^{1}(\mathbb{CP}^{1},\mathcal{O}(-2))$.
\end{proof}

\subsection{On conformal invariance}

The delta functions $\Delta_{m}^{(3)}(\bar{Z},\bar{X})$
and $\Delta^{(2)}\left(\bar{Z}_{1},\bar{Z}_{2},\bar{Z}_{3}\right)$ are globally conformally invariant, because the integrand $\delta^{(4)}_{\mathbb{C}}\left(\bar{Z}_{\alpha}\right)$ in the definition equations (\ref{delta}) and $(\ref{2Delta})$ is globally conformally invariant ($SU(2,2)$ invariant). We can show this fact as follows:

\begin{theorem}
The delta function $\delta^{(4)}_{\mathbb{C}}\left(\bar{Z}_{\alpha}\right)$ on $\mathbb{C}^{4}$ is invariant under the $SU(2,2)$ transformation:
\begin{align}
\delta^{(4)}_{\mathbb{C}}\left({\Lambda_{\alpha}}^{\beta} \bar{Z}_{\beta}\right) 
= \delta^{(4)}_{\mathbb{C}} \left(\bar{Z}_{\alpha}\right), \quad \left( {\Lambda_{\alpha}}^{\beta} \in SU(2,2) \right).
\end{align}
\label{GSU22delta}
\end{theorem}

\begin{proof}
From Theorem $\ref{thAz}$ and $\det \left(\Lambda_{\alpha}{}^{\beta}\right)=1$, we obtain
\begin{align}
\delta^{(4)}_{\mathbb{C}}\left({\Lambda_{\alpha}}^{\beta} \bar{Z}_{\beta}\right) 
= \dfrac{1}{\det \left( {\Lambda_{\alpha}}^{\beta} \right)}
\delta^{(4)}_{\mathbb{C}} \left(\bar{Z}_{\alpha}\right)
= \delta^{(4)}_{\mathbb{C}} \left(\bar{Z}_{\alpha}\right).
\notag
\end{align}
\end{proof}

In addition, the delta function $\delta^{(4)}_{\mathbb{C}}\left(\bar{Z}_{\alpha}\right)$ 
is invariant under the local conformal transformation because it is annihilated by all $SU(2,2)$ generators. Here, we represent the $SU(2,2)$ generators in terms of the dual twistor variables $\bar{Z}_{\alpha}=\left(\tilde{\pi}_{A}, \tilde{\omega}^{A^{\prime}}\right)$. 
The Lorentz generators $J_{A^{\prime} B^{\prime}}$, $\tilde{J}_{AB}$, the momentum operators $P_{AA^{\prime}}$, the dilatation operator $D$, and the special conformal generators $K_{AA^{\prime}}$ are represented as follows, by carrying out the Fourier transform of the $SU(2,2)$ generators in the momentum space \cite{Wit} to the dual complex twistor space:
\begin{align}
J_{A^{\prime} B^{\prime}} &= \dfrac{i}{2} 
\left( \tilde{\omega}_{A^{\prime}} \dfrac{\partial}{\partial \tilde{\omega}^{B^{\prime}}}
+ \tilde{\omega}_{B^{\prime}} \dfrac{\partial}{\partial \tilde{\omega}^{A^{\prime}}} \right),
\label{J}
\\
\tilde{J}_{A B} &= \dfrac{i}{2} 
\left( \tilde{\pi}_{A} \dfrac{\partial}{\partial \tilde{\pi}^{B}}
+\tilde{\pi}_{B} \dfrac{\partial}{\partial \tilde{\pi}^{A}} \right),
\\
P_{A A^{\prime}} &=
\tilde{\pi}_{A} \dfrac{\partial}{\partial \tilde{\omega}^{A^{\prime}}},
\\
D &= \dfrac{i}{2} \left( 
- \tilde{\omega}^{A^{\prime}} \dfrac{\partial}{\partial \tilde{\omega}^{A^{\prime}}}
+ \tilde{\pi}_{A} \dfrac{\partial}{\partial \tilde{\pi}_{A}}
\right),
\\
K_{A A^{\prime}} &=
\tilde{\omega}_{A^{\prime}} \dfrac{\partial}{\partial \tilde{\pi}^{A}}
.
\label{K}
\end{align}

\begin{theorem}
The delta function $\delta^{(4)}_{\mathbb{C}}\left(\bar{Z}_{\alpha}\right)$ on $\mathbb{C}^{4}$ is annihilated by all $SU(2,2)$ generators $J_{A^{\prime} B^{\prime}}$, $\tilde{J}_{AB}$, $P_{AA^{\prime}}$, $D$ and $K_{AA^{\prime}}$.
\label{SU22delta}
\end{theorem}

\begin{proof}
First, by applying the Lorentz generator $J_{0^{\prime}0^{\prime}}$ to $\delta^{(4)}_{\mathbb{C}}\left(\bar{Z}_{\alpha}\right)$, we have
\begin{align}
J_{{0}^{\prime} {0}^{\prime}} \delta^{(4)}_{\mathbb{C}}\left(\bar{Z}_{\alpha}\right)
=\dfrac{i}{(2 \pi i)^{4}}
\dfrac{1}{\tilde{\pi}_{0} \tilde{\pi}_{1}
\left(\tilde{\omega}^{0^{\prime}}\right)^{2}}.
\label{Jdelta}
\end{align}
This is equivalent to $0$ as an element of $H^{3}\left(\mathbb{CP}^{3}, \mathcal{O}(-4) \right)$ because the delta function $\delta^{(4)}_{\mathbb{C}}\left(\bar{Z}_{\alpha}\right)$ is a representative element of $H^{3}\left(\mathbb{CP}^{3}, \mathcal{O}(-4) \right)$.
Similarly, the actions of other Lorentz generators are equivalent to $0$.
Next, by applying 
the momentum operator $P_{00^{\prime}}$ to 
$\delta^{(4)}_{\mathbb{C}}\left(\bar{Z}_{\alpha}\right)$, we have
\begin{align}
P_{0 {0}^{\prime}} \delta^{(4)}_{\mathbb{C}}\left(\bar{Z}_{\alpha}\right)
= \dfrac{-1}{(2\pi i)^{4}}
\dfrac{1}{\tilde{\pi}_{1} \left(\tilde{\omega}^{0^{\prime}}\right)^{2} \tilde{\omega}^{1^{\prime}}}.
\label{Pdelta}
\end{align}
This is equivalent to $0$ as an element of $H^{3}\left(\mathbb{CP}^{3}, \mathcal{O}(-4) \right)$.
Similarly, the actions of other momentum operators are equivalent to $0$.
Furthermore, the application of dilatation operator $D$ to 
$\delta^{(4)}_{\mathbb{C}}\left(\bar{Z}_{\alpha}\right)$ equals to $0$ by direct calculation: $D \delta^{(4)}_{\mathbb{C}}\left(\bar{Z}_{\alpha}\right)=0$.
Finally, by applying the special conformal generator $K_{00^{\prime}}$ to $\delta^{(4)}_{\mathbb{C}}\left(\bar{Z}_{\alpha}\right)$, we have
\begin{align}
K_{0 0^{\prime}} \delta^{(4)}_{\mathbb{C}}\left(\bar{Z}_{\alpha}\right)
= \dfrac{1}{(2 \pi i)^{4}} \dfrac{1}{
\tilde{\pi}_{0} \left(\tilde{\pi}_{1}\right)^{2}
\tilde{\omega}^{0^{\prime}}}.
\end{align}
This is equivalent to $0$ as an element of $H^{3}\left(\mathbb{CP}^{3}, \mathcal{O}(-4) \right)$.
Similarly, the actions of other special conformal generators are equivalent to $0$.
Thus, each of the $SU(2,2)$ generators annihilates $\delta^{(4)}_{\mathbb{C}}\left(\bar{Z}_{\alpha}\right)$.
\end{proof}

From Theorems \ref{GSU22delta} and \ref{SU22delta}, we can see that $\Delta^{(3)}_{m}\left(\bar{Z},\bar{X}\right)$ defined by Eq. $(\ref{delta})$ and $\Delta^{(2)}\left(\bar{Z}_{1},\bar{Z}_{2},\bar{Z}_{3}\right)$ defined by Eq. (\ref{2Delta}) are globally and locally conformally invariant.

\subsection{New definition of the delta functions on the dual complex super twistor space}

In this subsection, we propose a new definition of the delta functions for superspaces.

\begin{definition}
First, let $\left(\bar{Z}_{\alpha},\xi_{k}\right)$ ( $\alpha=0,1,2,3$ ; $k=1,2,3,4$ ) be the coordinates of point 
$\mathcal{\bar{Z}}$
in complex supermanifold $\mathbb{C}^{4|4}$.
We define the delta function on $\mathbb{C}^{4|4}$ as
 the product of the delta function $\delta^{(4)}_{\mathbb{C}}\left(\bar{Z}_{\alpha}\right)$ on $\mathbb{C}^{4}$ and the Grassmann delta function $\delta^{(4)}\left(\xi_{k}\right)$ as follows:
\begin{align}
\delta^{(4|4)}_{\mathbb{C}} \left(\mathcal{\bar{Z}}\right)
:=\delta^{(4)}_{\mathbb{C}}\left(\bar{Z}_{\alpha}\right) \delta^{(4)}(\xi_{k})
= \dfrac{1}{(2 \pi i)^{4}}
\dfrac{1}{\bar{Z}_{0}\bar{Z}_{1}\bar{Z}_{2}\bar{Z}_{3}} \xi_{1} \xi_{2} \xi_{3} \xi_{4}.
\label{delta44}
\end{align}
\end{definition}

\begin{definition}
We define the delta function on the dual complex super twistor space $\mathbb{CP}^{3|4}$ as follows:
\begin{align}
\Delta^{(3|4)} \left( \mathcal{\bar{Z}}_{1}, \mathcal{\bar{Z}}_{2} \right)
:= \int_{C} \dfrac{dv}{v} \delta^{(4|4)}_{\mathbb{C}}
\left( \mathcal{\bar{Z}}_{1} - v \mathcal{\bar{Z}}_{2} \right),
\label{delta12}
\end{align}
where the integral contour $C$ surrounds only one of the four singular points of $\delta^{(4)}_{\mathbb{C}}\left(\bar{Z}_{1\alpha}-v\bar{Z}_{2\alpha}\right)$
in the same manner as in Eq. (\ref{delta}).
\end{definition}

\begin{definition}
We define the delta function which enforces the collinearity of the three points $\bar{\mathcal{Z}}_{1}$, $\bar{\mathcal{Z}}_{2}$, and $\bar{\mathcal{Z}}_{3}$ in $\mathbb{CP}^{3|4}$ as follows:
\begin{align}
\Delta^{(2|4)}\left(\mathcal{\bar{Z}}_{1},\mathcal{\bar{Z}}_{2},\mathcal{\bar{Z}}_{3}\right)
:= \int_{\gamma} \dfrac{dv}{v} \dfrac{dw}{w}
\delta^{(4|4)}_{\mathbb{C}} \left(\mathcal{\bar{Z}}_{1}-v\mathcal{\bar{Z}}_{2}-w\mathcal{\bar{Z}}_{3}\right),
\label{s2Delta}
\end{align}
where the integral contour $\gamma$ surrounds the point $(v,w)$ in $\mathbb{C}^{2}$ which satisfies 
Eq. $(\ref{gamma})$
in the same manner as in Eq. $(\ref{2Delta})$.
\end{definition}

\begin{proposition}
By carrying out the integration in Eq. $(\ref{s2Delta})$, we have
\begin{align}
\Delta^{(2|4)}\left(\mathcal{\bar{Z}}_{1},\mathcal{\bar{Z}}_{2},\mathcal{\bar{Z}}_{3}\right)
= 
&\dfrac{1}{
\langle 1 2 
\rangle \langle 2 3 \rangle
\langle 3 1 \rangle
}
\delta^{(2)}_{\mathbb{C}}\left(
\langle 2 3 \rangle 
{\tilde{\omega}_{1}}^{A^{\prime}}
+\langle 3 1 \rangle 
{\tilde{\omega}_{2}}^{A^{\prime}}
+\langle 1 2 \rangle 
{\tilde{\omega}_{3}}^{A^{\prime}}\right)
\notag
\\
&\times
\delta^{(4)}\left(
\langle 2 3 \rangle \xi_{1k}
+\langle 3 1 \rangle \xi_{2k}
+\langle 1 2 \rangle \xi_{3k} \right).
\label{Delta24123}
\end{align}
\end{proposition}

\begin{proof}
To carry out the integration, we transform the integral variables $v$ and $w$ to $a$ and $b$ in the same manner as Eq. $(\ref{ab})$. The integral contour $\gamma$ in Eq. $(\ref{s2Delta})$ is the $2$-dimensional torus which is the direct product of the circles expressed by Eqs. $(\ref{gamma0})$ and $(\ref{gamma1})$. By using the relations expressed in Eq. $(\ref{vwab})$, we can carry out the integration in Eq. $(\ref{s2Delta})$ in terms of variables $a$ and $b$ as follows:
\begin{align}
&\Delta^{(2|4)}\left(\mathcal{\bar{Z}}_{1},\mathcal{\bar{
Z}}_{2},\mathcal{\bar{Z}}_{3}\right)
\notag 
\\
=
&\dfrac{\langle 2 3 \rangle}{(2 \pi i)^{2}}
\int_{\gamma} da \wedge db
\dfrac{1}{(a- \bar{Z}_{10})(b- \bar{Z}_{11})}
\dfrac{1}{(a \bar{Z}_{31}-b \bar{Z}_{30})(-a \bar{Z}_{21}+b \bar{Z}_{20})}
\notag
\\
&
\,\,\,\,\,
\times
\prod_{{\alpha}=2}^{3}
\dfrac{1}{(2\pi i)}
\dfrac{1}
{\bar{Z}_{1\alpha}
-\dfrac{a \bar{Z}_{31}-b \bar{Z}_{30}}{\langle 2 3 \rangle} \bar{Z}_{2\alpha}
-\dfrac{-a \bar{Z}_{21}+b \bar{Z}_{20}}{\langle 2 3 \rangle} \bar{Z}_{3\alpha}}
\notag
\\
&
\,\,\,\,\,
\times
\prod_{k=1}^{4}
\left(
\xi_{1k}
-\dfrac{a \bar{Z}_{31}-b \bar{Z}_{30}}{\langle 2 3 \rangle} {\xi_{2}}_{k}
-\dfrac{-a \bar{Z}_{21}+b \bar{Z}_{20}}{\langle 2 3 \rangle}{\xi_{3}}_{k} \right)
\\
=
&\dfrac{\langle 2 3 \rangle}
{
\langle 1 2 \rangle
\langle 3 1 \rangle 
}
\prod_{\alpha=2}^{3}
\dfrac{1}{(2 \pi i)}
\dfrac{1}{
\bar{Z}_{1\alpha}
+\dfrac{\langle 3 1 \rangle}{\langle 2 3 \rangle} \bar{Z}_{2\alpha}
+\dfrac{\langle 1 2 \rangle}{\langle 2 3 \rangle} \bar{Z}_{3\alpha}
}
\notag
\\
&
\,\,\,\,\,
\times
\prod_{k=1}^{4} \left(
\xi_{1k} 
+\dfrac{\langle 3 1 \rangle}{\langle 2 3 \rangle} \xi_{2k}
+\dfrac{\langle 1 2 \rangle}{\langle 2 3 \rangle} \xi_{3k}
\right)
\\
=
&
\dfrac{1}{
\langle 1 2 \rangle 
\langle 2 3 \rangle
\langle 3 1 \rangle
}
\delta^{(2)}_{\mathbb{C}}\left(
\langle 2 3 \rangle {\tilde{\omega}_{1}}^{A^{\prime}}
+\langle 3 1 \rangle {\tilde{\omega}_{2}}^{A^{\prime}}
+\langle 1 2 \rangle {\tilde{\omega}_{3}}^{A^{\prime}}
\right)
\notag
\\
& \,\,\,\,\,
\times 
\delta^{(4)}\left(
\langle 2 3 \rangle \xi_{1k}
+\langle 3 1 \rangle \xi_{2k}
+\langle 1 2 \rangle \xi_{3k}
\right).
\label{101}
\end{align}
Thus, we obtain Eq. $(\ref{Delta24123})$. 
\end{proof}

From Eq. $(\ref{101})$, we see that 
$\Delta^{(2|4)}\left(\mathcal{\bar{Z}}_{1},\mathcal{\bar{
Z}}_{2},\mathcal{\bar{Z}}_{3}\right)$
is invariant under the scale transformation
$\left(\mathcal{\bar{Z}}_{1},\mathcal{\bar{Z}}_{2},\mathcal{\bar{Z}}_{3}\right) \rightarrow \left(c_{1}\mathcal{\bar{Z}}_{1}, c_{2}\mathcal{\bar{Z}}_{2},c_{3}\mathcal{\bar{Z}}_{3}\right)$.
In addition, 
$\Delta^{(2|4)}\left(\mathcal{\bar{Z}}_{1},\mathcal{\bar{
Z}}_{2},\mathcal{\bar{Z}}_{3}\right)$
has cyclic symmetry and is antisymmetric when any two variables are exchanged. These properties are exactly the nature of the MHV amplitude omitted the colour factor. 
In the next subsection, we show that $\Delta^{(2|4)}\left(\mathcal{\bar{Z}}_{1},\mathcal{\bar{
Z}}_{2},\mathcal{\bar{Z}}_{3}\right)$
is the three-particle MHV amplitude in the dual complex super twistor space.

The delta functions $\Delta^{(3|4)}\left(\bar{\mathcal{Z}}_{1}, \bar{\mathcal{Z}}_{2}\right)$
and $\Delta^{(2|4)}\left(\mathcal{\bar{Z}}_{1},\mathcal{\bar{
Z}}_{2},\mathcal{\bar{Z}}_{3}\right)$
are globally superconformally invariant, because the integrand $\delta^{(4|4)}_{\mathbb{C}}\left(\bar{\mathcal{Z}}\right)$
in the definition equations $(\ref{delta12})$ and $(\ref{s2Delta})$ 
is globally superconformally invariant ($PSU(2,2|4)$ invariant).
This can be expressed as follows:

\begin{theorem}
The delta function $\delta^{(4|4)}_{\mathbb{C}}\left(\mathcal{\bar{Z}}\right)$ on $\mathbb{C}^{4|4}$ defined by Eq. (\ref{delta44}) is invariant under the $PSU(2,2|4)$ transformation:
\begin{align}
\delta^{(4|4)}_{\mathbb{C}}\left(\Lambda \mathcal{\bar{Z}}\right)
=\delta^{(4|4)}_{\mathbb{C}}\left(\mathcal{\bar{Z}}\right),
\,\,\,\,\,
(\Lambda \in PSU(2,2|4)).
\end{align} 
\label{GPSU224delta}  
\end{theorem}

\begin{proof}
First, by integrating the product of $\delta^{(4|4)}_{\mathbb{C}}\left(\mathcal{\bar{Z}}\right)$
and an arbitrary function $f\left(\mathcal{\bar{Z}}\right)$ on $\mathbb{C}^{4|4}$ using the integration measure
\begin{align}
\Omega = d\bar{Z}_{0}d\bar{Z}_{1}d\bar{Z}_{2}d\bar{Z}_{3}d\xi_{1}d\xi_{2}d\xi_{3}d\xi_{4},
\label{measure}
\end{align}
we have
\begin{align}
\int \Omega \delta^{(4|4)}_{\mathbb{C}} \left(\mathcal{\bar{Z}}\right)
f\left(\mathcal{\bar{Z}}\right) 
&= \int \Omega \dfrac{1}{(2 \pi i)^{4}} \dfrac{1}{\bar{Z}_{0}\bar{Z}_{1}\bar{Z}_{2}\bar{Z}_{3}} \xi_{1} \xi_{2} \xi_{3} \xi_{4} f\left(\bar{Z}_{\alpha},\xi_{k}\right)
\notag
\\
&= \int d\xi_{1}d\xi_{2}d\xi_{3}d\xi_{4} \xi_{1} \xi_{2} \xi_{3} \xi_{4} f(0,\xi_{k})
\notag
\\
&= f(0,0),
\label{integral}
\end{align}
where the integration contour of each variable $\bar{Z}_{\alpha}$ surrounds singular point $\bar{Z}_{\alpha}=0$.

Next, we calculate the integration of the product of $\delta^{(4|4)}_{\mathbb{C}}\left(\Lambda \mathcal{\bar{Z}}\right)\,
(\Lambda \in PSU(2,2|4))$ and an arbitrary function $f\left(\mathcal{\bar{Z}}\right)$
to demonstrate that $\delta^{(4|4)}_{\mathbb{C}}\left(\Lambda \mathcal{\bar{Z}}\right)$
plays the same role as $\delta^{(4|4)}_{\mathbb{C}}
\left(\mathcal{\bar{Z}}\right)$. 
To carry out the integration, we perform a variable transformation as follows:
\begin{align}
\mathcal{\bar{X}} = \Lambda \mathcal{\bar{Z}} \, \Rightarrow \,
\mathcal{\bar{Z}} = \Lambda^{-1} \mathcal{\bar{X}}
= \left(
{C_{\alpha}}^{\beta} \bar{X}_{\beta} + {\Xi_{\alpha}}^{\ell} \psi_{\ell} ,\,
{\Theta_{k}}^{\beta} \bar{X}_{\beta} + {D_{k}}^{\ell} \psi_{\ell}
\right)
, \,\,\,\,\,
\mathcal{\bar{X}} := (\bar{X}_{\beta},\psi_{\ell}).
\end{align}
Here, $\Lambda^{-1}$ is a matrix defined by
\begin{align}
\Lambda^{-1} = \begin{pmatrix}
{C_{\alpha}}^{\beta} & {\Xi_{\alpha}}^{\ell} \\
{\Theta_{k}}^{\beta} & {D_{k}}^{\ell}
\end{pmatrix},
\end{align}
where $C$ and $D$ denote $4\times4$ Grassmann even matrices and $\Theta$ and $\Xi$ denote $4\times4$
Grassmann odd matrices.
In addition, $\Omega$ defined by Eq. $(\ref{measure})$ is invariant under the action of $PSU(2,2|4)$. Therefore, we have
\begin{align}
\Omega = d\bar{X}_{0}d\bar{X}_{1}d\bar{X}_{2}d\bar{X}_{3}
d\psi_{1} d\psi_{2} d\psi_{3} d\psi_{4}.
\end{align}
Using the variable $\mathcal{\bar{X}}$, the following integration can be performed:
\begin{align}
&\int \Omega \delta^{(4|4)}_{\mathbb{C}} \left(\Lambda \mathcal{\bar{Z}}\right) f\left(\mathcal{\bar{Z}}\right)
\notag
\\
=
& \int d\bar{X}_{0}d\bar{X}_{1}d\bar{X}_{2}d\bar{X}_{3}
d\psi_{1} d\psi_{2} d\psi_{3} d\psi_{4}
\delta^{(4|4)}_{\mathbb{C}} \left(\mathcal{\bar{X}}\right) 
f\left(\Lambda^{-1} \mathcal{\bar{X}}\right)
\notag
\\
=
& \int d\bar{X}_{0}d\bar{X}_{1}d\bar{X}_{2}d\bar{X}_{3}
d\psi_{1} d\psi_{2} d\psi_{3} d\psi_{4}
\notag
\\
&\times
\dfrac{1}{(2 \pi i)^{4}} \dfrac{1}{\bar{X}_{0}\bar{X}_{1}\bar{X}_{2}\bar{X}_{3}}
\psi_{1} \psi_{2} \psi_{3} \psi_{4} 
f\left(C\bar{X}+\Xi \psi, \Theta \bar{X} +D \psi \right)
\notag
\\
=
& \int d\psi_{1} d\psi_{2} d\psi_{3} d\psi_{4}
\psi_{1} \psi_{2} \psi_{3} \psi_{4} 
f(\Xi \psi, D \psi)
\notag
\\
=& f(0,0),
\end{align}
where the integration contour of each variable $\bar{X}_{\alpha}$ surrounds the singular point $\bar{X}_{\alpha}=0$. This result is the same as that obtained using Eq. (\ref{integral}). Therefore, $\delta^{(4|4)}_{\mathbb{C}} \left(\Lambda \mathcal{\bar{Z}}\right)$ is equal to 
$\delta^{(4|4)}_{\mathbb{C}} \left(\mathcal{\bar{Z}}\right)$.
\end{proof}

In addition, the delta function $\delta^{(4|4)}_{\mathbb{C}} \left(\mathcal{\bar{Z}}\right)$ is invariant under the locally superconformal transformation because it is annihilated by all $PSU(2,2|4)$ generators. 
Here, we represent the $PSU(2,2|4)$ generators in terms of the dual complex super twistor variables 
$\mathcal{\bar{Z}} = \left( \tilde{\pi}_{A}, \tilde{\omega}^{A^{\prime}}, \xi_{k} \right)$
by carrying out the Fourier transform of the $PSU(2,2|4)$ generators in the momentum superspace \cite{Wit}
to the dual complex super twistor space.
The generators in the dual complex super twistor space consist of the $SU(2,2)$ generators represented by Eqs. $(\ref{J}) \sim (\ref{K})$ and the following generators:
\begin{align}
&{R_{k}}^{\ell} =  -\xi_{k} \dfrac{\partial}{\partial \xi_{\ell}} 
+\dfrac{1}{4} \delta_{k}^{\ell} \sum_{m=1}^{4} \xi_{m} \dfrac{\partial}{\partial \xi_{m}},
\label{R}
\\
&{Q^{A^{\prime}}}_{k} = -\xi_{k} \dfrac{\partial}{\partial \tilde{\omega}_{A^{\prime}}},
\\
&{S_{A^{\prime}}}^{k} = \tilde{\omega}_{A^{\prime}}
\dfrac{\partial}{\partial {\xi_{k}}},
\\ 
&\tilde{Q}^{Ak} = \tilde{\pi}^{A} \dfrac{\partial}{\partial \xi_{k}},
\\
&\tilde{S}_{Ak} = \xi_{k} \dfrac{\partial}{\partial \tilde{\pi}^{A}}.
\label{tildeS}
\end{align}

\begin{theorem}
The delta function
$\delta^{(4|4)}_{\mathbb{C}}\left(\mathcal{\bar{Z}}\right)$
on $\mathbb{C}^{4|4}$
is annihilated by all $PSU(2,2|4)$ generators.
\label{PSU224delta}
\end{theorem}

\begin{proof}
The application of the $SU(2,2)$ generators to
$\delta^{(4|4)}_{\mathbb{C}}\left(\mathcal{\bar{Z}}\right)$
is equal to zero in the same manner as in the case of $\delta^{(4)}_{\mathbb{C}}\left(\bar{Z}_{\alpha}\right)$
shown in Theorem \ref{SU22delta}.
Therefore, we only investigate the application of generators ${R_{k}}^{\ell}$,  ${Q^{A^{\prime}}}_{k}$, ${S_{A^{\prime}}}^{k}$,
$\tilde{Q}^{Ak}$, and $\tilde{S}_{Ak}$.
First, we see that the application of ${R_{k}}^{\ell}$ to $\delta^{(4)}\left(\xi_{k}\right)$ is equal to zero through direct calculations. Therefore, ${R_{k}}^{\ell}\delta^{(4|4)}_{\mathbb{C}}\left(\mathcal{\bar{Z}}\right)=0$.
Next, we have ${Q^{A^{\prime}}}_{k} \delta^{(4|4)}_{\mathbb{C}}\left(\mathcal{\bar{Z}}\right)=0$ and 
$\tilde{S}_{Ak} \delta^{(4|4)}_{\mathbb{C}}\left(\mathcal{\bar{Z}}\right)=0$ using the property of the delta function $\xi_{k}\delta^{(4)}\left(\xi_{k}\right)=0$.
Furthermore, we have ${S_{A^{\prime}}}^{k}\delta^{(4|4)}_{\mathbb{C}}\left(\mathcal{\bar{Z}}\right)=0$ using the property of the delta function
$\tilde{\omega}^{A^{\prime}}\delta^{(4)}_{\mathbb{C}}\left(\bar{Z}_{\alpha}\right)=0$.
Similarly, we have $\tilde{Q}^{Ak} \delta^{(4|4)}_{\mathbb{C}}\left(\mathcal{\bar{Z}}\right)=0$ using the property of the delta function $\tilde{\pi}_{A}\delta^{(4)}_{\mathbb{C}}\left(\bar{Z}_{\alpha}\right)=0$.
Thus, each $PSU(2,2|4)$ generator annihilates $\delta^{(4|4)}_{\mathbb{C}}\left(\mathcal{\bar{Z}}\right)$.
\end{proof}

From Theorems \ref{GPSU224delta} and \ref{PSU224delta}, we can see that
$\Delta^{(3|4)}\left(\bar{\mathcal{Z}}_{1}, \bar{\mathcal{Z}}_{2}\right)$
defined by Eq. $(\ref{delta12})$
and 
$\Delta^{(2|4)}\left(\mathcal{\bar{Z}}_{1},\mathcal{\bar{
Z}}_{2},\mathcal{\bar{Z}}_{3}\right)$
defined by Eq. $(\ref{s2Delta})$ 
are globally and locally superconformally invariant.

\subsection{Relationship with the scattering amplitudes in the momentum superspace}

In this subsection, we demonstrate that the inverse Fourier transform of the delta function 
$\Delta^{(2|4)}(\bar{\mathcal{Z}}_{1},\bar{\mathcal{Z}}_{2},\bar{\mathcal{Z}}_{3})$ is the three-particle MHV amplitude for $\mathcal{N}=4$ in the momentum superspace. In addition, the inverse Fourier transform of
 the product of two delta functions
$\Delta^{(3|4)}(\bar{\mathcal{Z}}_{1},\bar{\mathcal{Z}}_{2})\Delta^{(3|4)}(\bar{\mathcal{Z}}_{1},\bar{\mathcal{Z}}_{3})$
is the three-particle $\overline{\text{MHV}}$ amplitude.
Hence, when the twistor space is complex space, the amplitudes do not have sign factors and are superconformally invariant. 
Here, the inverse Fourier transform from dual complex super twistor space to momentum superspace is given by
\begin{align}
\mathcal{F}^{-1} \left[ f \left(\tilde{\pi}_{A}, \tilde{\omega}^{A^{\prime}}, \xi_{k} \right) \right]
= \int d^{2} \tilde{\omega} d^{4} \xi
e^{-\tilde{\omega}^{A^{\prime}} \pi_{A^{\prime}} - \xi_{k} \eta^{k}} f \left(\tilde{\pi}_{A}, \tilde{\omega}^{A^{\prime}}, \xi_{k} \right),
\end{align}
where $(\pi_{A^{\prime}}, \tilde{\pi}_{A}, \eta^{k})$ denotes the coordinates of the momentum superspace.
Detailed properties of this inverse Fourier transform
are provided in literature \cite{JN}.

\begin{remark}
In this subsection, we define  
$[ ij ] := \pi_{iA^{\prime}} \pi_{j}^{A^{\prime}}
= \pi_{i 0^\prime} \pi_{j 1^\prime} - \pi_{i 1^\prime} \pi_{j 0^\prime}$. 
In addition, we consider that $\tilde{\pi}_{i}$ is not the complex conjugate of $\pi_{i}$ (the momenta are complex in this case) because the three-particle MHV and $\overline{\text{MHV}}$ amplitudes are not when $\tilde{\pi}_{i}$ is the complex conjugate of $\pi_{i}$ (the momenta are real in this case) \cite{ElHu}.   
Therefore, $[ij]$ is not the complex conjugate of $\langle ij \rangle$.
\end{remark}

\begin{theorem}
The inverse Fourier transform of the delta function on the dual complex super twistor space defined by Eq. $(\ref{s2Delta})$ is
\begin{align}
\mathcal{F}^{-1} \left[ \Delta^{(2|4)} \left( \mathcal{\bar{Z}}_{1}, \mathcal{\bar{Z}}_{2}, \mathcal{\bar{Z}}_{3}
\right) \right]
=
\dfrac{1}{ 
\langle 12 \rangle 
\langle 23 \rangle
\langle 31 \rangle
}
\delta^{(4)}_{\mathbb{C}} \left( \sum_{i=1}^{3} \tilde{\pi}_{i A} 
\pi_{i A^{\prime}} \right)
\delta^{(8)} \left( \sum_{i=1}^{3}  
\tilde{\pi}_{iA} \eta_{i}^{k} \right).
\label{finverseDelta24123}
\end{align}
This is exactly the three-particle MHV amplitude for $\mathcal{N}=4$ SYM in the momentum superspace.
\end{theorem}

\begin{proof}
First, the inverse Fourier transform of the integrand in the right hand side of Eq. $(\ref{s2Delta})$ is
\begin{align}
& \quad \mathcal{F}^{-1} \left[ \delta^{(4|4)}_{\mathbb{C}}
\left( \mathcal{\bar{Z}}_{1} - v \mathcal{\bar{Z}}_{2}
- w \mathcal{\bar{Z}}_{3} \right) \right]
\notag \\
&
= \int \prod_{j=1}^{3} d^{2} \tilde{\omega}_{j} d^{4} \xi_{j}
e^{-\tilde{\omega}_{j}^{A^{\prime}} \pi_{j A^{\prime}}
- \xi_{jk} \eta_{j}^{k} }
\delta^{(4|4)}_{\mathbb{C}}
\left( \mathcal{\bar{Z}}_{1} - v \mathcal{\bar{Z}}_{2}
- w \mathcal{\bar{Z}}_{3} \right)
\notag \\
&
= \prod_{A=0}^{1} \dfrac{1}{2\pi i}
\dfrac{1}{\left( \tilde{\pi}_{1A} - v \tilde{\pi}_{2A}
- w \tilde{\pi}_{3A} \right)}
\int d^{2} \tilde{\omega}_{2} d^{2} \tilde{\omega}_{3}
e^{- \tilde{\omega}_{2}^{A^\prime} (v \pi_{1 A^\prime} + \pi_{2 A^\prime})}
e^{- \tilde{\omega}_{3}^{A^\prime} (w \pi_{1 A^\prime} + \pi_{3 A^\prime})}
\notag
\\
& \quad \times
\delta^{(4)} \left( v \eta_{1}^{k} + \eta_{2}^{k} \right)
\delta^{(4)} \left( w \eta_{1}^{k} + \eta_{3}^{k} \right),
\label{finverse1}
\end{align}
where the integral contour of $\tilde{\omega}_{1}^{A^{\prime}}$ surrounds the singular point 
$\tilde{\omega}_{1}^{A^{\prime}}=v\tilde{\omega}_{2}^{A^{\prime}}+w\tilde{\omega}_{3}^{A^{\prime}}$.
We set the integration domain in Eq. $(\ref{finverse1})$
as  
$\left\{ \text{Re}\, \tilde{\omega}_{2}^{A^{\prime}} > 0,\,  \text{Re}\, \tilde{\omega}_{3}^{A^{\prime}} > 0 \right\}$.
By using the integration formula
\begin{align}
\int_{0}^{\infty} e^{-zx} dx = \dfrac{1}{z}, \quad
\text{Re} z > 0, 
\label{intformula}
\end{align}
we can carry out the integration
of the variables $\tilde{\omega}_{2}^{A^\prime}$
and $\tilde{\omega}_{3}^{A^\prime}$
in Eq. $(\ref{finverse1})$ as follows:
\begin{align}
&\int d^{2} \tilde{\omega}_{2} d^{2} \tilde{\omega}_{3}
e^{- \tilde{\omega}_{2}^{A^\prime} (v \pi_{1 A^\prime} + \pi_{2 A^\prime})}
e^{- \tilde{\omega}_{3}^{A^\prime} (w \pi_{1 A^\prime} + \pi_{3 A^\prime})}
\notag
\\
=
& \prod_{A^{\prime}=0^{\prime}}^{1^{\prime}}
\dfrac{1}{v \pi_{1 A^\prime} + \pi_{2 A^\prime}}
\prod_{B^{\prime}=0^{\prime}}^{1^{\prime}}
\dfrac{1}{w \pi_{1 B^\prime} + \pi_{3 B^\prime}},
\label{omegaint}
\end{align}
where 
$\left\{\text{Re}(v \pi_{1 A^\prime} + \pi_{2 A^\prime})>0,\,\text{Re}(w \pi_{1 A^\prime} + \pi_{3 A^\prime})>0
\right\}$.
Here, we extend the 
$\left\{ \left( 
\pi_{1A^{\prime}},\, \pi_{2A^{\prime}},\, \pi_{3A^{\prime}}
\right)\right\}$
region in Eq. $(\ref{omegaint})$
 to the direct product of
$\left\{v \pi_{1 A^\prime} + \pi_{2 A^\prime} \ne 0
\right\}$
and
$\left\{w \pi_{1 A^\prime} + \pi_{3 A^\prime} \ne 0
\right\}$
by using the analytic continuation method. From Eqs. $(\ref{s2Delta})$, $(\ref{finverse1})$, and
$(\ref{omegaint})$, we have
\begin{align}
& \quad \mathcal{F}^{-1} \left[ \Delta^{(2|4)} \left( \mathcal{\bar{Z}}_{1}, \mathcal{\bar{Z}}_{2},\mathcal{\bar{Z}}_{3} \right) \right]
\notag \\
&= \int_{\gamma} \dfrac{dv}{v} \dfrac{dw}{w}
\mathcal{F}^{-1} \left[ \delta^{(4|4)}_{\mathbb{C}}
\left( \mathcal{\bar{Z}}_{1} - v \mathcal{\bar{Z}}_{2}
- w \mathcal{\bar{Z}}_{3} \right) \right]
\notag \\
&= \int_{\gamma} \dfrac{dv}{v} \dfrac{dw}{w}
\prod_{A=0}^{1} \dfrac{1}{2\pi i}
\dfrac{1}{\left( \tilde{\pi}_{1A} - v \tilde{\pi}_{2A}
- w \tilde{\pi}_{3A} \right)}
\prod_{A^{\prime}=0^{\prime}}^{1^{\prime}}
\dfrac{1}{v \pi_{1 A^\prime} + \pi_{2 A^\prime}}
\prod_{B^{\prime}=0^{\prime}}^{1^{\prime}}
\dfrac{1}{w \pi_{1 B^\prime} + \pi_{3 B^\prime}}
\notag \\
& \quad \times
\delta^{(4)} \left( v \eta_{1}^{k} + \eta_{2}^{k} \right)
\delta^{(4)} \left( w \eta_{1}^{k} + \eta_{3}^{k} \right).
\label{finverse1.5}
\end{align}
To carry out this integration, we transform the integral variables $v$ and $w$ to $a$ and $b$ in the same manner as Eq. $(\ref{ab})$ and set the integral contour $\gamma$ to surround only the singular points $a=\tilde{\pi}_{10}$ and $b=\tilde{\pi}_{11}$ in the same manner as in Eqs. $(\ref{gamma0})$ and $(\ref{gamma1})$. By performing integration, we obtain 
\begin{align}
(\ref{finverse1.5})
= \dfrac{1}{
\langle 31 \rangle 
\langle 12 \rangle
\langle 23 \rangle^{3}
}
\delta^{(2)}_{\mathbb{C}} \left( \tilde{\pi}_{3A} P^{A}{}_{A^\prime} \right)
\delta^{(2)}_{\mathbb{C}} \left( \tilde{\pi}_{2A} P^{A}{}_{A^\prime}\right)
\delta^{(4)} \left( \tilde{\pi}_{3 A} H^{Ak} \right)  
\delta^{(4)} \left( \tilde{\pi}_{2 A} H^{Ak} \right), 
\label{finverse2}
\end{align}
where
\begin{align}
P^{A}{}_{A^\prime} := \tilde{\pi}_{1}^{A} \pi_{1 A^{\prime}} 
+ \tilde{\pi}_{2}^{A} \pi_{2 A^{\prime}} + \tilde{\pi}_{3}^{A} \pi_{3 A^{\prime}}, \quad
H^{Ak} := \tilde{\pi}_{1}^{A} \eta_{1}^{k} + \tilde{\pi}_{2}^{A} \eta_{2}^{k} + \tilde{\pi}_{3}^{A} \eta_{3}^{k}.
\end{align}
Here, using Theorem \ref{thAz}, the product of the delta functions on $\mathbb{C}^{2}$ in Eq. $(\ref{finverse2})$ can be expressed as follows:
\begin{align}
&\delta^{(2)}_{\mathbb{C}} \left( \tilde{\pi}_{3A} P^{A}{}_{A^\prime} \right)
\delta^{(2)}_{\mathbb{C}} \left( \tilde{\pi}_{2A} P^{A}{}_{A^\prime}\right)
\notag
\\
=& \dfrac{1}{\left( \tilde{\pi}_{3 A} P^{A}{}_{0^\prime} \right)\left( \tilde{\pi}_{2 A} P^{A}{}_{0^\prime} \right)}
\cdot
\dfrac{1}{\left( \tilde{\pi}_{3 A} P^{A}{}_{1^\prime} \right)\left( \tilde{\pi}_{2 A} P^{A}{}_{1^\prime} \right)}
\notag \\
=& \dfrac{1}{\langle 32 \rangle P^{0}{}_{0^\prime}P^{1}{}_{0^\prime}}
\cdot
\dfrac{1}{\langle 32 \rangle P^{0}{}_{1^\prime}P^{1}{}_{1^\prime}}
= \dfrac{1}{\langle 32 \rangle^2} \prod_{A=0}^{1} 
\prod_{A^{\prime}=0^{\prime}}^{1^\prime}
\dfrac{1}{P_{AA^{\prime}}}
\notag \\
=& \dfrac{1}{\langle 32 \rangle^2}
\delta^{(4)}_{\mathbb{C}} \left( \tilde{\pi}_{1A} \pi_{1 A^{\prime}} 
+ \tilde{\pi}_{2A} \pi_{2 A^{\prime}} + \tilde{\pi}_{3A} \pi_{3 A^{\prime}}\right).
\label{dcpopip}
\end{align}
Furthermore, the product of the Grassmann delta functions in Eq. $(\ref{finverse2})$ can be represented as follows:
\begin{align}
&\delta^{(4)} \left( \tilde{\pi}_{3 A} H^{Ak} \right)  
\delta^{(4)} \left( \tilde{\pi}_{2 A} H^{Ak} \right) 
\notag \\
=&\prod_{k=1}^{4} 
\left( \tilde{\pi}_{30} H^{0k} + \tilde{\pi}_{31} H^{1k}\right)
\left( \tilde{\pi}_{20} H^{0k} + \tilde{\pi}_{21} H^{1k}\right)
\notag \\
=&\prod_{k=0}^{4}
\langle 32 \rangle H^{0k} H^{1k}
= \langle 32 \rangle^{4} 
\delta^{(4)}\left( {H_{0}}^{k} \right)
\delta^{(4)}\left( {H_{1}}^{k} \right)
\notag \\
=& \langle 32 \rangle^{4}
\delta^{(8)} \left( \tilde{\pi}_{1A} \eta_{1}^{k} + \tilde{\pi}_{2A} \eta_{2}^{k} + \tilde{\pi}_{3A} \eta_{3}^{k} \right).
\label{d4pih}
\end{align}
Substituting Eqs. $(\ref{dcpopip})$ and $(\ref{d4pih})$ into Eq. $(\ref{finverse2})$, we have
\begin{align}
(\ref{finverse2}) 
= 
\dfrac{
\delta^{(4)}_{\mathbb{C}} \left( \tilde{\pi}_{1A} \pi_{1 A^{\prime}} 
+ \tilde{\pi}_{2A} \pi_{2 A^{\prime}} + \tilde{\pi}_{3A} \pi_{3 A^{\prime}}\right)
\delta^{(8)} \left( \tilde{\pi}_{1A} \eta_{1}^{k} + \tilde{\pi}_{2A} \eta_{2}^{k} + \tilde{\pi}_{3A} \eta_{3}^{k} \right)
}
{
\langle 12 \rangle 
\langle 23 \rangle
\langle 31 \rangle
}.
\end{align}
Thus, we obtain Eq. $(\ref{finverseDelta24123})$.
\end{proof}

\begin{theorem}
The inverse Fourier transform of the product of the two delta functions on the dual complex super twistor space
 $\Delta^{(3|4)}\left(\mathcal{\bar{Z}}_{1}, \mathcal{\bar{Z}}_{2}\right)$ and  
$\Delta^{(3|4)}\left(\mathcal{\bar{Z}}_{1}, \mathcal{\bar{Z}}_{3}\right)$ is 
\begin{align}
& 
\mathcal{F}^{-1} \left[
\dfrac{1}{(2 \pi i)^{2}}
\Delta^{(3|4)}\left( \mathcal{\bar{Z}}_{1}, \mathcal{\bar{Z}}_{2} \right)
\Delta^{(3|4)}\left( \mathcal{\bar{Z}}_{1}, \mathcal{\bar{Z}}_{3} \right)
\right]
\notag \\
=&\dfrac{1}{[12][23][31]}
\delta^{(4)}_{\mathbb{C}} \left( \tilde{\pi}_{1A} \pi_{1A^\prime} + \tilde{\pi}_{2A} \pi_{2A^\prime} + \tilde{\pi}_{3A} \pi_{3A^\prime} \right)
\delta^{(4)} \left( [23] \eta_{1}^{k} + [31] \eta_{2}^{k} + [12] \eta_{3}^{k} \right).
\label{finverse0}
\end{align}
This is exactly the three-particle ${\overline{\text{MHV}}}$ amplitude for $\mathcal{N}=4$ SYM in the momentum superspace.
\end{theorem}

\begin{proof}
First, from Eq. $(\ref{delta12})$, the product of 
$\Delta^{(3|4)}\left(\mathcal{\bar{Z}}_{1}, \mathcal{\bar{Z}}_{2}\right)$ and  
$\Delta^{(3|4)}\left(\mathcal{\bar{Z}}_{1}, \mathcal{\bar{Z}}_{3}\right)$ is expressed as follows:
\begin{align}
\Delta^{(3|4)}\left( \mathcal{\bar{Z}}_{1}, \mathcal{\bar{Z}}_{2} \right)
\Delta^{(3|4)}\left( \mathcal{\bar{Z}}_{1}, \mathcal{\bar{Z}}_{3} \right)
= \int \dfrac{dv}{v} \dfrac{dw}{w}
\delta^{(4|4)}_{\mathbb{C}} \left( \mathcal{\bar{Z}}_{1} - v \mathcal{\bar{Z}}_{2} \right) 
\delta^{(4|4)}_{\mathbb{C}} \left( \mathcal{\bar{Z}}_{1} - w \mathcal{\bar{Z}}_{3} \right). 
\label{delta1213}
\end{align}
The inverse Fourier transform of the integrand in the right hand side of Eq. $(\ref{delta1213})$ is
\begin{align}
& 
\mathcal{F}^{-1} \left[ \delta^{(4|4)}_{\mathbb{C}} \left( \mathcal{\bar{Z}}_{1} - v \mathcal{\bar{Z}}_{2} \right) 
\delta^{(4|4)}_{\mathbb{C}} \left( \mathcal{\bar{Z}}_{1} - w \mathcal{\bar{Z}}_{3} \right) \right]
\notag \\
=
& \int \prod_{j=1}^{3} d^{2} \tilde{\omega}_{j} d^{4} \xi_{j}
e^{- \tilde{\omega}_{j}^{A^\prime} \pi_{jA^\prime}
- \xi_{jk} \eta_{j}^{k} } 
\delta^{(4|4)}_{\mathbb{C}} \left( \mathcal{\bar{Z}}_{1} - v \mathcal{\bar{Z}}_{2} \right) 
\delta^{(4|4)}_{\mathbb{C}} \left( \mathcal{\bar{Z}}_{1} - w \mathcal{\bar{Z}}_{3} \right)
\notag \\
=
& \prod_{A=0}^{1} \dfrac{1}{(2 \pi i)}\dfrac{1}{\tilde{\pi}_{1A} - v \tilde{\pi}_{2A}}
\prod_{B=0}^{1} \dfrac{1}{(2 \pi i)}\dfrac{1}{\tilde{\pi}_{1B} - w \tilde{\pi}_{3B}}
\prod_{k=1}^{4} \left( vw \eta_{1}^{k} + w \eta_{2}^{k} + v \eta_{3}^{k} \right)
\notag \\
& \times
\int \prod_{j=1}^{3} d^{2} \tilde{\omega}_{j} 
e^{- \tilde{\omega}_{j}^{A^\prime} \pi_{jA^\prime}}
\prod_{A^{\prime}=0^{\prime}}^{1^{\prime}}
\dfrac{1}{(2 \pi i)}\dfrac{1}{\tilde{\omega}_{1}^{A^\prime}-v \tilde{\omega}_{2}^{A^\prime}}
\prod_{B^{\prime}=0^{\prime}}^{1^{\prime}}
\dfrac{1}{(2 \pi i)}\dfrac{1}{\tilde{\omega}_{1}^{B^\prime}-w \tilde{\omega}_{3}^{B^\prime}}.
\label{finverse3}
\end{align}
Here, we set the integral contour of $\tilde{\omega}_{1}$ 
to surround only singular points of ${1}/\left(\tilde{\omega}_{1}^{A^{\prime}}-v\tilde{\omega}_{2}^{A^{\prime}}\right)$. Therefore, for the integration of  
the variables $\tilde{\omega}_{j}$ in Eq. $(\ref{finverse3})$, we have
\begin{align}
& 
\int \prod_{j=1}^{3} d^{2} \tilde{\omega}_{j} 
e^{- \tilde{\omega}_{j}^{A^\prime} \pi_{jA^\prime}}
\prod_{A^{\prime}=0^{\prime}}^{1^{\prime}}
\dfrac{1}{(2 \pi i)}\dfrac{1}{\tilde{\omega}_{1}^{A^\prime}-v \tilde{\omega}_{2}^{A^\prime}}
\prod_{B^{\prime}=0^{\prime}}^{1^{\prime}}
\dfrac{1}{(2 \pi i)}\dfrac{1}{\tilde{\omega}_{1}^{B^\prime}-w \tilde{\omega}_{3}^{B^\prime}}
\notag \\
= 
&\int d^2 \tilde{\omega}_{2} d^2 \tilde{\omega}_{3}
e^{-v \tilde{\omega}_{2}^{A^\prime} \pi_{1A^\prime}-\tilde{\omega}_{2}^{A^\prime} \pi_{2A^\prime}
-\tilde{\omega}_{3}^{A^\prime} \pi_{3A^\prime}}
\prod_{B^\prime=0^\prime}^{1^\prime} \dfrac{1}{(2 \pi i)}
\dfrac{1}{v \tilde{\omega}_{2}^{B^\prime} - w \tilde{\omega}_{3}^{B^\prime}}.
\label{omegaintegral}
\end{align}
Furthermore, we set the integral contour of $\tilde{\omega}_{2}$ to surround the singular points of
$1/\left(v\tilde{\omega}_{2}^{B^{\prime}} - w\tilde{\omega}_{3}^{B^{\prime}} \right)$
and the integral contour of $\tilde{\omega}_{3}$ to be 
$\{\text{Re}\tilde{\omega}_{3}^{A^{\prime}}>0\}$.
Therefore, by using the integration formula $(\ref{intformula})$ , we have 
\begin{align}
(\ref{omegaintegral})
= 
& \dfrac{1}{v^{2}} \int d^2 \tilde{\omega}_{3}
\exp \left[-\tilde{\omega}_{3}^{A^\prime} \left( w \pi_{1A^\prime} + \dfrac{w}{v} \pi_{2A^\prime} + \pi_{3A^\prime} \right) \right]
\notag \\
= 
& (2 \pi i)^2 \delta^{(2)}_{\mathbb{C}} 
\left( vw \pi_{1A^\prime} + w \pi_{2A^\prime} + v \pi_{3A^\prime} \right),
\label{omegaint2}
\end{align}
where 
$\left\{\text{Re}\left( w\pi_{1A^{\prime}} + \dfrac{w}{v} \pi_{2A^{\prime}} + \pi_{3A^{\prime}} \right)>0\right\}$. 
Here, we extend the $\{(\pi_{1A^{\prime}},\pi_{2A^{\prime}},\pi_{3A^{\prime}})\}$
region in Eq. $(\ref{omegaint2})$ to 
$\{vw \pi_{1A^\prime} + w \pi_{2A^\prime} + v \pi_{3A^\prime} \ne 0\}$ by using the analytic continuation method.
Substituting Eq. $(\ref{omegaint2})$ into Eq. $(\ref{finverse3})$, we have
\begin{align}
(\ref{finverse3}) 
=
& 
(2 \pi i)^{2}
\delta^{(2)}_{\mathbb{C}} \left( \tilde{\pi}_{1A} - v \tilde{\pi}_{2A} \right)
\delta^{(2)}_{\mathbb{C}} \left( \tilde{\pi}_{1B} - w \tilde{\pi}_{3B} \right)
\delta^{(4)} \left( vw \eta_{1}^{k} + w \eta_{2}^{k} + v \eta_{3}^{k} \right)
\notag \\
& \times
\delta^{(2)}_{\mathbb{C}} \left( vw \pi_{1A^\prime} + w \pi_{2A^\prime} + v \pi_{3A^\prime} \right).
\label{ifofintegrand}
\end{align}
From Eqs. $(\ref{delta1213})$ and $(\ref{ifofintegrand})$, we have
\begin{align}
& 
\mathcal{F}^{-1} \left[  \Delta^{(3|4)}\left( \mathcal{\bar{Z}}_{1}, \mathcal{\bar{Z}}_{2} \right)
\Delta^{(3|4)}\left( \mathcal{\bar{Z}}_{1}, \mathcal{\bar{Z}}_{3} \right) \right]
\notag \\ 
= 
&\int \dfrac{dv}{v} \dfrac{dw}{w} 
\mathcal{F}^{-1} \left[ \delta^{(4|4)}_{\mathbb{C}} \left( \mathcal{\bar{Z}}_{1} - v \mathcal{\bar{Z}}_{2} \right) 
\delta^{(4|4)}_{\mathbb{C}} \left( \mathcal{\bar{Z}}_{1} - w \mathcal{\bar{Z}}_{3} \right) \right]
\notag \\
= 
&\dfrac{(2 \pi i)^2}{\tilde{\pi}_{20} \tilde{\pi}_{30}}
\delta_{\mathbb{C}} (\langle 21 \rangle)
\delta_{\mathbb{C}} (\langle 31 \rangle)
\notag \\
& \times
\delta^{(2)}_{\mathbb{C}} 
\left( \tilde{\pi}_{10} \pi_{1A^\prime} 
+ \tilde{\pi}_{20} \pi_{2A^\prime} 
+ \tilde{\pi}_{30} \pi_{3A^\prime} \right)
\delta^{(4)} \left( \tilde{\pi}_{10} \eta_{1}^{k}  
+ \tilde{\pi}_{20} \eta_{2}^{k} + \tilde{\pi}_{30} \eta_{3}^{k}
\right),
\label{finverse4}
\end{align}
where the integral contour of $v$ is set to surround only the singular point of $1/(\tilde{\pi}_{10}-v\tilde{\pi}_{20})$
and that of $w$ is set to surround only the  
singular point of $1/(\tilde{\pi}_{10}-w\tilde{\pi}_{30})$.
Here, from the third delta function in Eq. $(\ref{finverse4})$, we have
\begin{align}
\tilde{\pi}_{10} \pi_{1A^\prime} 
+ \tilde{\pi}_{20} \pi_{2A^\prime} 
+ \tilde{\pi}_{30} \pi_{3A^\prime} = 0.
\label{barpipi}
\end{align}
By using Eq. $(\ref{barpipi})$, 
Theorem $\ref{thAz}$
and $\tilde{\pi}_{1A} \tilde{\pi}_{1}^{A}=0$, the product of the first and second delta functions in Eq. $(\ref{finverse4})$ can be expressed as follows:
\begin{align}
& \delta_{\mathbb{C}} (\langle 21 \rangle)
\delta_{\mathbb{C}} (\langle 31 \rangle)
\notag \\
= 
&[23] \delta^{(2)}_{\mathbb{C}} \left( 
\pi_{2A^\prime} \langle 21 \rangle
+ \pi_{3A^\prime} \langle 31 \rangle \right)
\notag \\
= 
&[23] \delta^{(2)}_{\mathbb{C}} \left( 
( \pi_{1A^\prime} \tilde{\pi}_{1A} 
+ \pi_{2A^\prime} \tilde{\pi}_{2A} 
+ \pi_{3A^\prime} \tilde{\pi}_{3A} ) \tilde{\pi}_{1}^{A}
\right)
\notag \\
= 
&[23] \delta^{(2)}_{\mathbb{C}} \left( 
( \pi_{1A^\prime} \tilde{\pi}_{11} 
+ \pi_{2A^\prime} \tilde{\pi}_{21} 
+ \pi_{3A^\prime} \tilde{\pi}_{31} ) \tilde{\pi}_{1}^{1}
\right)
\notag \\
= 
&[23] \dfrac{1}{(\tilde{\pi}_{10})^2}
\delta^{(2)}_{\mathbb{C}} \left( 
 \pi_{1A^\prime} \tilde{\pi}_{11} 
+ \pi_{2A^\prime} \tilde{\pi}_{21} 
+ \pi_{3A^\prime} \tilde{\pi}_{31} 
\right).
\label{2131}
\end{align}
Furthermore, by contracting Eq. $(\ref{barpipi})$ with 
$\pi_{2}$ or $\pi_{3}$, we have
\begin{align}
\tilde{\pi}_{30} = \dfrac{[12]}{[23]} \tilde{\pi}_{10}, \quad
\tilde{\pi}_{20} = \dfrac{[31]}{[23]} \tilde{\pi}_{10}.
\label{pi32}
\end{align}
By using Eq. $(\ref{pi32})$, the fourth delta function in Eq. $(\ref{finverse4})$ can be expressed as follows:
\begin{align}
& 
\delta^{(4)} \left( \tilde{\pi}_{10} \eta_{1}^{k}  
+ \tilde{\pi}_{20} \eta_{2}^{k} + \tilde{\pi}_{30} \eta_{3}^{k}
\right)
\notag \\
=
&\delta^{(4)} \left( \tilde{\pi}_{10} \eta_{1}^{k}  
+ \dfrac{[31]}{[23]} \tilde{\pi}_{10} \eta_{2}^{k} + \dfrac{[12]}{[23]} \tilde{\pi}_{10} \eta_{3}^{k}
\right)
\notag \\
= 
&\dfrac{(\tilde{\pi}_{10})^4}{[23]^4}
\delta^{(4)} \left( [23] \eta_{1}^{k}  
+ [31] \eta_{2}^{k} + [12] \eta_{3}^{k} \right).
\label{eta123}
\end{align}
Substituting Eqs. $(\ref{2131})$, $(\ref{pi32})$, and
$(\ref{eta123})$ into Eq. $(\ref{finverse4})$, we have
\begin{align}
(\ref{finverse4})
= 
\dfrac{(2 \pi i)^{2}}{[12][23][31]}
\delta^{(4)}_{\mathbb{C}} \left( 
\tilde{\pi}_{1A} \pi_{1A^\prime}  
+ \tilde{\pi}_{2A} \pi_{2A^\prime}
+ \tilde{\pi}_{3A} \pi_{3A^\prime}
\right)
\delta^{(4)} \left( [23] \eta_{1}^{k} + [31] \eta_{2}^{k}
+ [12] \eta_{3}^{k} \right).
\end{align}
Thus, we obtain Eq. $(\ref{finverse0})$. 
\end{proof}

\section{Summary and discussion}

We revealed a new geometrical interpretation of the sign factors which arise in the real twistor space. In addition, we proposed a new definition of the delta functions on the dual complex (super) twistor space in terms of the \v{C}ech cohomology group.

We demonstrated that the domain of the tilded $\delta$-function $\tilde{\delta}^{(3)}_{-n-4}(W,Y)$ without the sign factor is the three-dimensional sphere $S^{3} \approx \mathbb{RP}^{3} \times O(1)$, by comparing with the delta function on the real twistor space $\mathbb{RP}^{3}$, ${\delta}^{(3)}_{-n-4}(W,Y)$, with the sign factor.
Furthermore, we demonstrated that the domain of the collinear $\tilde{\delta}$-function which enforces the collinearity of the three twistors in $\mathbb{RP}^{3}$,
$\tilde{\delta}^{(2)}(W_{1},W_{2},W_{3})$,
without the sign factors is $G_{2,4}(\mathbb{R}) \times O(1)$, and that of the collinear $\delta$-function, 
${\delta}^{(2)}(W_{1},W_{2},W_{3})$, with the sign factors is 
$G_{2,4}(\mathbb{R})$.
Thus, we revealed that the sign factors play a role in dividing the domain of the delta functions by $O(1)$.

We proposed a new definition of the delta function on  the $n$-dimensional complex space $\mathbb{C}^{n}$,
denoted by $\delta^{(n)}_{\mathbb{C}}(z_{1},z_{2},\cdots,z_{n})$, as a representative element of the $(n-1)$-th \v{C}ech cohomology group on $\mathbb{CP}^{n-1}$ with coefficients in the sheaf $\mathcal{O}(-n)$, denoted by $H^{n-1}(\mathbb{CP}^{n-1},\mathcal{O}(-n))$.
By using this delta function to impose the equivalence relation of the complex projective space $\mathbb{CP}^{3}$ for any two points $\bar{X}$ and $\bar{Z}$, we defined the delta function on the dual complex twistor space $\mathbb{PT}^{*}=\mathbb{CP}^{3}$, denoted by  
$\Delta^{(3)}_{m}\left(\bar{Z},\bar{X}\right)$, without any sign factor. We demonstrated that this function possesses the property of the delta function on $\mathbb{PT}^{*}$ and is a representative element of the 
\v{C}ech cohomology group $H^{2}\left(\mathbb{CP}^{2},\mathcal{O}(-3)\right)$. Furthermore, we defined the delta function $\Delta^{(2)}\left(\bar{Z}_{1},\bar{Z}_{2},\bar{Z}_{3}\right)$ which enforces the collinearity of the three twistors $\bar{Z}_{1}$, $\bar{Z}_{2}$, and $\bar{Z}_{3}$ in $\mathbb{PT}^{*}$ without any sign factors, as a representative element of the  
\v{C}ech cohomology group $H^{1}\left(\mathbb{CP}^{1},\mathcal{O}(-2)\right)$. We also showed that both $\Delta^{(3)}_{m}\left(\bar{Z},\bar{X}\right)$ and $\Delta^{(2)}\left(\bar{Z}_{1},\bar{Z}_{2},\bar{Z}_{3}\right)$
are invariant under local and global conformal transformations ($SU(2,2)$ transformation). 
This invariance results from the fact that
they do not have sign factors and their ingredients' delta function on $\mathbb{C}^{4}$, denoted by $\delta^{(4)}_\mathbb{C} \left(\bar{Z}_{\alpha}\right)$, is invariant under the local and global conformal transformations.
Similarly, we defined the delta function
$\Delta^{(3|4)}\left(\mathcal{\bar{Z}}_{1}, \mathcal{\bar{Z}}_{2}\right)$
on the dual complex super twistor space $\mathbb{CP}^{3|4}$
and the delta function $\Delta^{(2|4)}\left(\mathcal{\bar{Z}}_{1}, \mathcal{\bar{Z}}_{2}, \mathcal{\bar{Z}}_{3}\right)$ which enforces the collinearity of the three super twistors in $\mathbb{CP}^{3|4}$.
These delta functions are invariant under local and global super conformal transformations  ($PSU(2,2|4)$ transformation).
Furthermore, we also showed that the inverse Fourier transform of $\Delta^{(2|4)}\left(\mathcal{\bar{Z}}_{1}, \mathcal{\bar{Z}}_{2}, \mathcal{\bar{Z}}_{3}\right)$
is the three-particle MHV amplitude for $\mathcal{N}=4$
SYM in the momentum superspace. Similarly, the inverse Fourier transform of the product of the two delta functions
$\Delta^{(3|4)}\left(\mathcal{\bar{Z}}_{1}, \mathcal{\bar{Z}}_{2}\right)$ and
$\Delta^{(3|4)}\left(\mathcal{\bar{Z}}_{1}, \mathcal{\bar{Z}}_{3}\right)$
is the three-particle $\overline{\text{MHV}}$ amplitude for $\mathcal{N}=4$
SYM in the momentum superspace.
For this reason, the delta functions $\Delta^{(3|4)}(\bar{\mathcal{Z}}_{1},\bar{\mathcal{Z}}_{2})$ and $\Delta^{(2|4)}(\bar{\mathcal{Z}}_{1},\bar{\mathcal{Z}}_{2},\bar{\mathcal{Z}}_{3})$ are scattering amplitudes for $\mathcal{N}=4$ SYM in the dual complex super twistor space.

The delta functions $\delta^{(3)}_{-n-4}(W,Y)$ and $\tilde{\delta}^{(3)}_{-n-4}(W,Y)$ are conformally invariant solutions of the differential equation of the twistor wave function \cite{Pen5}, because they possess a degree of homogeneity $(-n-4)$ in $W$ and $n$ in $Y$. For the same reason, the delta function $\Delta^{(3)}_{m}\left(\bar{Z},\bar{X}\right)$ is a conformally invariant solution of the differential equation of the twistor wave function. Therefore, these delta functions resemble the invariant function or the Feynman propagator which is well known in quantum field theory in the Minkowski space. 
Here, the delta function on twistor space is considered to be a fundamental object than the $\overline{\text{MHV}}$ amplitude in momentum space, because the $\overline{\text{MHV}}$ amplitude is constructed from the product of the two delta functions, as seen from Eqs $(\ref{tpdfists})$ and $(\ref{finverse0})$. For this reason, quantum field theory in twistor space may be more fundamental than that in space-time. In future studies, we will construct an operator formalism of quantum field theory in the real and complex twistor spaces. This would reveal the physical meaning of the sign factors. As a conjecture, the sign factors may correspond to the metric of the Fock space of the quantum field theory in twistor space. When the twistor space is real space, the metric of the Fock space may be indefinite, because the three-particle MHV and $\overline{\text{MHV}}$ amplitudes are the product of the sign factors and the delta functions. However, when the twistor space is complex space, the metric of the Fock space may be definite, because the three-particle MHV and $\overline{\text{MHV}}$ amplitudes do not have the sign factors.  Furthermore, the operator formalism of quantum field theory in the complex twistor space may be related to conformal gravity, because the delta functions, which are the ingredients of scattering amplitudes, are conformally invariant. 
Additionally, investigating the BCFW recursion relation in the complex twistor space in terms of the \v{C}ech cohomology group is a topic for future research.

\end{document}